\documentclass{article}

\usepackage{preamble}

\usepackage{amsmath}
\usepackage{amssymb}
\usepackage{graphicx}   
\usepackage{xcolor}
\usepackage{blkarray}
\usepackage{tikz, tikz-cd}
\usepackage{wrapfig}
\usepackage[width=.75\textwidth, font=small]{caption}
\usepackage{subcaption}
\usepackage{ifthen}
\usepackage[ruled,vlined,linesnumbered]{algorithm2e}
\SetAlFnt{\footnotesize}                         
\SetAlCapFnt{\footnotesize}\SetAlCapNameFnt{\footnotesize}  
\SetArgSty{textup}                               

\SetCommentSty{algcmt}                           
\SetKw{Reject}{reject $B$}

\newcommand{\GX}{G_{B_{\mathsf{X}}}}
\newcommand*{\boldone}{\text{\usefont{U}{bbold}{m}{n}1}}
\newcommand*{\boldzero}{\text{\usefont{U}{bbold}{m}{n}0}}
\providecommand{\ord}{}
\renewcommand{\ord}{\operatorname{ord}}

\title{Graphical Analysis of Lifted Product Code Constructions}
\author{Ragnar Freij-Hollanti, Kirsten D. Morris, Patricija \v{S}apokait\.{e}}

\begin{document}

\maketitle

\begin{abstract}
Lifted product codes are an important family of quantum low-density parity-check (QLDPC) codes, as they were the first QLDPC code family shown to be asymptotically good. Understanding the structure of their parity-check matrices $H_{\mathsf{X}}$ and $H_{\mathsf{Z}}$, as well as the associated Tanner graphs, is essential for analyzing their decoding behavior and error-floor performance. We show that the Tanner graphs associated to $H_{\mathsf{X}}$ and $H_{\mathsf{Z}}$ are isomorphic, establish conditions ensuring their connectivity, and characterize and count the minimal absorbing sets of the underlying base graph, together with how they lift. We further classify the short tailless backtrackless closed walks in the base graph, which yields explicit conditions on the base matrix under which the Tanner graphs have girth $8$. These results provide new insight into the combinatorial structures influencing decoding performance. Finally, we use these conditions to enumerate the labelings of $3\times 4$
base matrices that satisfy them for small lift sizes, and we show by
simulation that codes built this way have better decoding performance under BP+OSD than a code with a random labeling.
\end{abstract}

\section{Introduction}
Quantum computing has significant promise for solving problems that are intractable for classical systems, but fully realizing this potential requires robust quantum error correction. Constructing quantum error correcting code families that are efficient and practically implementable remains an active area of research.

Quantum low-density parity-check (LDPC) codes are particularly promising candidates for quantum error correction, especially given Gottesman's result proving constant overhead of fault-tolerant quantum computation when using quantum LDPC codes with a constant encoding rate~\cite{Gottesman-overhead}. Despite this promise, constructing quantum LDPC codes that are asymptotically good while also maintaining desirable decoding performance in the finite block length regime has been nontrivial.

Panteleev and Kalachev resolved this open question with the introduction of lifted product (LP) codes, the first construction of asymptotically good quantum LDPC codes~\cite{panteleev2022asymptotically}. These codes combine the protograph based structure from classical LDPC codes along with tensoring operations, yielding codes with both constant rate and growing distance. However, a detailed understanding of the structural and combinatorial properties of these codes, particularly explicit constructions, remains open.

Graph lifts have been extensively studied for classical LDPC codes, particularly in terms of their impact on connectivity, girth, and cycle structure~\cite{mitchell2014quasi, smarandache2025structural, dolecek2014non}. However, in the context of lifted product codes, the interplay between graph lifting and tensoring operations introduces new phenomena that are not yet fully characterized. In particular, the relationship between the base protograph and the resulting Tanner graph requires more care, influencing the analysis of structural properties such as connectivity, girth, and the presence of small substructures that influence decoding performance~\cite{mao2001heuristic, halford2005algorithm, karimi2012efficient}.

In this work, we study the lifted product codes derived from a single base matrix $B$ as in~\cite{raveendran2025minimumdistancesfinitelengthlifted}. We show that the Tanner graphs of $H_{\mathsf{X}}$ and $H_{\mathsf{Z}}$ are isomorphic and give conditions on the labels of $B$ under which these graphs are connected. We classify and count the minimal absorbing sets of the base graph $G_{B_{\mathsf{X}}}$ and describe how its $(2,0)$ absorbing sets lift to $G_{H_{\mathsf{X}}}$. Finally, we classify the tailless backtrackless closed (TBC) walks of length $4$, $6$, and $8$ in $G_{B_{\mathsf{X}}}$, which yields conditions on $B$ under which $G_{H_{\mathsf{X}}}$ has girth $8$.

This paper is organized as follows. Section~\ref{section:preliminaries} provides background, Section~\ref{section:connectivity} gives the connectivity criteria, Section~\ref{section:girth_absorbing} analyzes absorbing sets, Section~\ref{section:cycle_analysis_B} classifies short TBC walks and derives girth conditions on the base matrix, and Section \ref{section:simulations} presents simulations for explicit codes constructed using our mathematical criteria.

\section{Preliminaries}\label{section:preliminaries}

\subsection{Graph Theory Tools}
Let $n$ be a natural number and let $[n]$ denote the set $\{1, 2, \dots, n\}$. We write $\boldzero_n$ and $\boldone_n$ to denote the $n\times n$ all-zeros and all-ones matrices, respectively, dropping the subscript $n$ when clear from context. We denote the $n \times n$ identity matrix as $I_n$. A convenient way to represent a parity check matrix $H$ is through its \textit{Tanner graph} $G_H$ \cite{T81}.

\begin{definition}
For a binary linear code $\mathcal{C} \subseteq \mathbb{F}_2^n$ and a parity check representation $H \in \mathbb{F}_2^{m \times n}$ for $\mathcal{C}$ we associate a bipartite graph $G_H=(V,C;E)$, called the \textit{Tanner graph} of $\mathcal{C}$ corresponding to $H$. The vertices of the Tanner graph are the \textit{variable nodes} $V=\{v_1,\dots, v_n\}$ and the \textit{check nodes} $C=\{c_1,\dots, c_m\}$ and the edge set is given by $c_i \sim v_j$ if and only if the entry $h_{ij}$ in $H$ is nonzero.
\end{definition}

We will make repeated use of \emph{graph homomorphisms}. A graph homomorphism from a graph $G=(V,E)$ to $G'=(V',E')$ is a map $\phi:V\to V'$ that preserves edges, so if $\{u,v\}\in E$, then $\{\phi(u),\phi(v)\}\in E'$. We denote the girth of a graph $G$ by $g(G)$. An important combinatorial invariant of a Tanner graph is its \textit{absorbing sets}. These were first introduced by Dolecek et al. as a combinatorial structure whose presence is detrimental to the error floor of classical decoders~\cite{Dolecek07}, later shown to also be harmful for syndrome-based bit-flipping decoders in the quantum setting~\cite{Morris2026AMC}.

\begin{definition}\label{def:absorbing}
Given a Tanner graph $G$, an \textit{$(a,b)$ absorbing set $\mathcal{A}$} is a subset of $a=|\mathcal{A}|$ variable nodes such that there are $b$ odd degree check nodes in the induced subgraph $G_{\mathcal{A} \cup \mathcal{N}(\mathcal{A})}$ and every variable node $v \in \mathcal{A}$ has strictly more even than odd degree neighbors in $G_{\mathcal{A} \cup \mathcal{N}(\mathcal{A})}$.
\end{definition}

Next, we establish necessary notation used to define our main objects of study, the \emph{lifted product codes}, introduced in~\cite{panteleev2022asymptotically}.

\begin{definition}
We denote the group of \textit{$r \times r$ circulant permutations} as $C_r=\{I, P, P^2, \dots, P^{r-1}\}$ where $P$ is the cyclic permutation mapping $i$ to $i+1 \mod r$.
\end{definition}

With a slight abuse of notation, we refer to the permutation matrices and their corresponding permutations interchangeably, with the particular object clear from the context. In particular,
\[
P=\begin{pmatrix}
    0&0&\cdots &0&1\\
    1&0&\cdots&0& 0\\
    \vdots&\ddots&\ddots&&\vdots\\
    0&\cdots &0&1&0
\end{pmatrix}.
\]

Consider a graph $G=(V,E)$ with vertex set $V=\{v_1,\dots,v_n\}$, called the \emph{protograph}, and let $r\in \mathbb{N}$. A \emph{lift of degree $r$} (or simply an \emph{$r$-lift}) of $G$ is a graph $\tilde{G}$ with vertex set $\tilde{V}=\{v_{i,k} : v_i \in V,\; k \in [r]\}$, together with a projection map $\pi : \tilde{G} \to G$, defined on vertices by $\pi(v_{i,k})=v_i$, such that for each edge $e=v_iv_j \in E$, the edges $\pi^{-1}(e)$ of the lifted graph form a perfect matching between the fibers
\[
\pi^{-1}(v_i)=\{v_{i,1},\dots,v_{i,r}\}
\quad \text{and} \quad
\pi^{-1}(v_j)=\{v_{j,1},\dots,v_{j,r}\}.
\]

We will be concerned with a particular permutation-based $r$-lift, constructed by assigning to each directed edge $e=v_iv_j \in E$ a circulant permutation $\tau_e \in C_r$. The lifted edges corresponding to $e$ are then
\begin{equation}\label{eq:twist}
v_{i,k}v_{j,\tau_e(k)},
\qquad 1 \leq k \leq r.
\end{equation}

We refer to the bi-adjacency matrix associated to the protograph as the \emph{base matrix} and the bi-adjacency matrix associated with the lifted graph as the \emph{lifted matrix}. We will frequently have labeled edges in the protograph $G$, in which case the base matrix has these labels as entries.

Edges are treated as directed when assigning permutations. When the protograph $G$ is a Tanner graph, we orient all edges from variable nodes to check nodes. That is, given nodes $v_j$ and $c_i$ with edge label $b_{ij}$, we associate a traversal from $v_j$ to $c_i$ with the permutation $b_{ij}$, and a traversal from $c_i$ to $v_j$ with the permutation $b_{ij}^{-1}$.

A \emph{walk} $W$ of length $t$ in a graph $G$ is a homomorphism to $G$ from the path with $t$ edges. Concretely, it is a sequence $w_0, w_1, \dots, w_t$ of vertices of $G$ such that $w_i w_{i+1}$ is an edge for all $0\leq i\leq t-1$. The walk is \emph{closed} if $w_0=w_t$. If the traversal from $w_i$ to $w_{i+1}$ is associated with the permutation $\sigma_i$, the \emph{permutation product} of $W$ is $\omega(W)=\prod_{i=0}^{t-1} \sigma_i$.

In a Tanner graph, the vertices of a closed walk $W$ alternate between variable nodes and check nodes. Assuming without loss of generality that $w_0$ (and therefore every $w_i$ with $i$ even) is a variable node, and that the edge $w_{i}w_{i+1}$ is labeled $\tau_i$, we have $\omega(W)=\tau_0\tau_1^{-1}\tau_2\cdots\tau_{t-1}^{-1}$.

There is a rich history of results relating the graph properties of a protograph to those of its lifts. We state two of them here for reference.

\begin{theorem}\label{thm:gross_tucker1}\cite{gross2001topological}
Let $\gamma$ be a $k$-cycle in a protograph $G$ such that the permutation product $\omega(\gamma)$ has order $s$ in $C_r$. Then each component of the preimage $\pi^{-1}(\gamma)$ is a $ks$-cycle, and there are $\frac{|C_r|}{s}=\frac{r}{s}$ such components.
\end{theorem}

Given a connected protograph $G$ with edges labeled by elements of $C_r$, let $T$ be a spanning tree of $G$. The \emph{fundamental cycles} of $G$ are formed by taking each edge $e\notin T$ and considering the unique cycle in $T \cup e$.

\begin{theorem}\label{thm:gross_tucker2}\cite{gross2001topological}
Let $G$ be a connected protograph. The number of connected components of $\tilde{G}$ is the index in $C_r$ of the subgroup generated by the permutation products around the fundamental cycles of $G$.
\end{theorem}

Crucially, $\tilde{G}$ is connected if and only if the permutation products around the fundamental cycles of $G$ generate all of~$C_r$.

\begin{example}\label{ex:lift}
Consider the base matrix $A$ with entries $P^{\varepsilon_{ij}} \in C_2$, $i, j \in[2]$, and its lifted matrix $A'$ given in~\eqref{eq:intro_example_lift_matrix}.
\begin{equation}\label{eq:intro_example_lift_matrix}
A=\begin{bmatrix}
    P^0 & P^1\\
    P^1 & P^0\\
\end{bmatrix}, \qquad A'=\begin{bmatrix}
    1 & 0 & 0 & 1\\
    0 & 1 & 1 & 0\\
    0 & 1 & 1 & 0\\
    1 & 0 & 0 & 1\\
\end{bmatrix}.
\end{equation}
We illustrate the protograph $G_A$ and its lift $\tilde{G}_A$ in Figure~\ref{fig:intro_example}. Consider the closed walk $W=v_1\sim c_1 \sim v_2 \sim c_2 \sim v_1$ in $G_A$. Its permutation product is $\omega(W)=P^0(P^1)^{-1}P^0(P^1)^{-1}=P^{-2}=I$ in $C_2$, which has order $s=1$. Hence, by Theorem~\ref{thm:gross_tucker1}, the $4$-cycle of $G_A$ lifts to $\frac{|C_2|}{s}=2$ disjoint $4$-cycles in $\tilde{G}_A$.
\end{example}

\begin{figure}[ht!]
    \centering
    \begin{subfigure}{0.48\linewidth}
        \centering
        \begin{tikzpicture}[thick,scale=.9]

\node[circle,draw=black,label={left:\small{$v_1$}}] (v1) at (-0.75,1){};
\node[circle,draw=black,label={left:\small{$v_2$}}] (v2) at (-0.75,-1){};

\node[shape=rectangle,draw=black,label={right:\small{$c_1$}}] (c1) at (0.75,1){};
\node[shape=rectangle,draw=black,label={right:\small{$c_2$}}] (c2) at (0.75,-1){};

\path[-,thick] (v1) edge node[pos=0.25,, above] {\small{$P^0$}} (c1);
\path[-,thick] (v2) edge node[pos=0.05, above] {\small{$P^1$}} (c1);

\path[-,thick] (v1) edge node[pos=0.05, below, yshift=-4] {\small{$P^1$}} (c2);
\path[-,thick] (v2) edge node[pos=0.25, below] {\small{$P^0$}} (c2);

\end{tikzpicture}
    \end{subfigure}
    \hfill
    \begin{subfigure}{0.48\linewidth}
        \centering
\begin{tikzpicture}[thick,scale=.8]

\node[circle,draw=black, label={left:\small{$v_{1,1}$}}] (v1) at (-1.5,1){};
\node[circle,draw=black, label={left:\small{$v_{1,2}$}}] (v2) at (-1.5,0.33){};
\node[circle,draw=black, label={left:\small{$v_{2,1}$}}] (v3) at (-1.5,-0.33){};
\node[circle,draw=black, label={left:\small{$v_{2,2}$}}] (v4) at (-1.5,-1){};

\node[shape=rectangle,draw=black, label={right:\small{$c_{1,1}$}}] (c1) at (1.5,1.5){};
\node[shape=rectangle,draw=black, label={right:\small{$c_{1,2}$}}] (c2) at (1.5,0.5){};
\node[shape=rectangle,draw=black, label={right:\small{$c_{2,1}$}}] (c3) at (1.5,-0.5){};
\node[shape=rectangle,draw=black, label={right:\small{$c_{2,2}$}}] (c4) at (1.5,-1.5){};

\path[-,thick] (v1) edge (c1);
\path[-,thick] (v4) edge (c1);

\path[-,thick] (v2) edge (c2);
\path[-,thick] (v3) edge (c2);

\path[-,thick] (v2) edge (c3);
\path[-,thick] (v3) edge (c3);

\path[-,thick] (v1) edge (c4);
\path[-,thick] (v4) edge (c4);

\end{tikzpicture}
    \end{subfigure}
    \vspace{0.5cm}
    \caption{A labeled protograph $G_A$ and its lift $\tilde{G}_A$ from Example~\ref{ex:lift}.}
    \label{fig:intro_example}
\end{figure}

\subsection{Lifted Product Code Construction}
Quantum low-density parity-check (QLDPC) codes are an important class of quantum error-correcting codes defined by having sparse parity check matrices, making them suitable for low complexity iterative decoding and fault-tolerant quantum computing. In this work, we focus on lifted product codes, which are a subclass of Calderbank--Shor--Steane (CSS) codes~\cite{CS96}, constructed from two classical linear codes. Consider two linear codes $\mathcal{C}_{\mathsf{X}}, \mathcal{C}_{\mathsf{Z}} \subseteq \mathbb{F}_2^n$ such that $\mathcal{C}_{\mathsf{X}}^{\perp} \subseteq \mathcal{C}_{\mathsf{Z}}$, with corresponding parity check matrices $H_{\mathsf{X}}$ and $H_{\mathsf{Z}}$. We construct a CSS code, denoted $\mathrm{CSS}(\mathcal{C}_{\mathsf{X}}, \mathcal{C}_{\mathsf{Z}})$, with parity check matrix $H=\begin{bmatrix}
    H_{\mathsf{X}} & \boldzero\\
    \boldzero & H_{\mathsf{Z}}\\
\end{bmatrix}.$ The condition that $\mathcal{C}_{\mathsf{X}}^{\perp} \subseteq \mathcal{C}_{\mathsf{Z}}$ implies that $H_{\mathsf{Z}}H_{\mathsf{X}}^T=\boldzero$, hence the stabilizer generators for the code commute and $\mathrm{CSS}(\mathcal{C}_{\mathsf{X}}, \mathcal{C}_{\mathsf{Z}})$ forms a valid \emph{stabilizer code}. In quantum communications channels, the $\mathsf{X}$ errors, or bit-flip errors, and $\mathsf{Z}$ errors, or phase-flip errors, are typically correlated, but we ignore this correlation for CSS codes~\cite{mackay2004sparse} and decode $\mathsf{X}$ and $\mathsf{Z}$ errors separately, using $H_{\mathsf{X}}$ to decode $\mathsf{Z}$ errors and $H_{\mathsf{Z}}$ to decode $\mathsf{X}$ errors.

For this reason, the performance of iterative graph-based decoders for these codes reduces to understanding the two Tanner graphs separately. \emph{Lifted product codes} refer to a broad CSS code construction, using base matrices and some group action to construct a CSS code with good parameters. We consider a particular instance as in~\cite{raveendran2025minimumdistancesfinitelengthlifted}, and we describe this construction below.

Consider an $m \times n$ matrix $B$ with entries $b_{ij}=P^{\varepsilon_{ij}}$, where $P$ is the generator of $C_r$. We then construct the two base matrices
\begin{equation}\label{eq:parity_matrices}
\begin{split}
B_{\mathsf{X}}&=[B \otimes I_n \: | \: I_m \otimes B^{*}]\\
\text{and} \qquad B_{\mathsf{Z}} &= [I_n \otimes B \: | \: B^* \otimes I_m],
\end{split}
\end{equation}

for $B^*$ the conjugate transpose of $B$. That is, if $B=[b_{ij}]_{i \in [m], j \in [n]}$, then $B^*$ is the $n\times m$ matrix whose $(j,i)$ entry is $b_{ij}^{-1}$, where $b_{ij}^{-1}=(P^{\varepsilon_{ij}})^{-1}=P^{r-\varepsilon_{ij}}$ is the inverse of $b_{ij}=P^{\varepsilon_{ij}}$ in $C_r$.
We perform the $r$-lift described in~\eqref{eq:twist} on the matrices $B_{\mathsf{X}}$ and $B_{\mathsf{Z}}$, and obtain the resulting matrices $H_{\mathsf{X}}$ and $H_{\mathsf{Z}}$. Given the dimensions $m\times n$ of
the base matrix $B$, we observe that both $H_{\mathsf{X}}$ and $H_{\mathsf{Z}}$ have $r(mn)$ rows and $r(m^2+n^2)$ columns. Since $H_{\mathsf{Z}}H_{\mathsf{X}}^T = \boldzero$, we obtain the resulting lifted product quantum CSS code $\mathcal{C}(H_{\mathsf{X}}, H_{\mathsf{Z}})$, which we denote as $\mathcal{C}$ with parameters $[[r(n^2+m^2), k^{\mathcal{C}}, d^{\mathcal{C}}]]$.

Observe that $G_B$ and $G_{B^*}$ have identical node sets (with the roles of check and variable nodes interchanged), but $G_B$ and $G_{B^*}$ have the edge label mappings $\ell_1$ and $\ell_2$, respectively, as seen in~\eqref{eq:edge_labels}.

\begin{equation}\label{eq:edge_labels}
\begin{aligned}
\ell_1:E_B &\to C_r,
& (c_i,v_j)&\mapsto b_{ij},\\
\ell_2:E_B &\to C_r,
& (c_i,v_j)&\mapsto b_{ij}^{-1}.
\end{aligned}
\end{equation}

Enumerate the check nodes of $G_{B_{\mathsf{X}}}$ as $c_{i,j}$, the variable nodes of $G_{B \otimes I_n}$ as $v_{k,j}^L$, and the variable nodes of $G_{I_m \otimes B^{*}}$ as $v_{i, \ell}^R$, for $i, \ell \in [m]$ and $j,k \in [n]$. We observe that $G_{B \otimes I_n}$ is a disjoint union of $n$ copies of $G_B$, denoted $G_1,\dots, G_n$, and $G_{I_m \otimes B^{*}}$ is a disjoint union of $m$ copies of $G_{B^*}$, denoted $G'_1,\dots, G'_m$.
To visualize these subgraphs from the matrix perspective, for a $2\times 3$ base matrix $B= \begin{tiny}
 \begin{bmatrix}b_1&b_2&b_3\\b_4& b_5&b_6\end{bmatrix}\end{tiny}$, we write out the terms of $B_{\mathsf{X}}
= \left[
B\otimes I_3
\;\middle|\;
I_2\otimes B^*
\right]$ in \eqref{eq:B_X}.

\begin{equation}\label{eq:B_X}
\setcounter{MaxMatrixCols}{20}
\renewcommand{\arraystretch}{1.2}
\begin{blockarray}{cccccccccccccc}
&
v_{1,1}^L & v_{1,2}^L & v_{1,3}^L
&
v_{2,1}^L & v_{2,2}^L & v_{2,3}^L
&
v_{3,1}^L & v_{3,2}^L & v_{3,3}^L
&
v_{1,1}^R & v_{1,2}^R
&
v_{2,1}^R & v_{2,2}^R\\
\begin{block}{c[ccccccccccccc]}
c_{1,1}
& b_{1} & 0 & 0
& b_{2} & 0 & 0
& b_{3} & 0 & 0
& b_{1}^{-1} & b_{4}^{-1}
& 0 & 0\\
c_{1,2}
& 0 & b_{1} & 0
& 0 & b_{2} & 0
& 0 & b_{3} & 0
& b_{2}^{-1} & b_{5}^{-1}
& 0 & 0\\
c_{1,3}
& 0 & 0 & b_{1}
& 0 & 0 & b_{2}
& 0 & 0 & b_{3}
& b_{3}^{-1} & b_{6}^{-1}
& 0 & 0\\
c_{2,1}
& b_{4} & 0 & 0
& b_{5} & 0 & 0
& b_{6} & 0 & 0
& 0 & 0
& b_{1}^{-1} & b_{4}^{-1}\\
c_{2,2}
& 0 & b_{4} & 0
& 0 & b_{5} & 0
& 0 & b_{6} & 0
& 0 & 0
& b_{2}^{-1} & b_{5}^{-1}\\
c_{2,3}
& 0 & 0 & b_{4}
& 0 & 0 & b_{5}
& 0 & 0 & b_{6}
& 0 & 0
& b_{3}^{-1} & b_{6}^{-1}\\
\end{block}
\end{blockarray}
\end{equation}

In terms of graphs, $G_{B_{\mathsf{X}}}$ thus decomposes into the $n$ disjoint copies $G_1,\dots, G_n$ of $G_B$ ``on the left'' and the $m$ disjoint copies $G_1',\dots, G_m'$ of $G_{B^*}$ ``on the right'', where the $i$th check node of $G_j$ is identified with the $j$th check node of $G'_i$, for $i\in[m]$, $j\in [n]$. These identified nodes $c_{i,j}$ serve as the check nodes of $G_{B_{\mathsf{X}}}$, as in Figure~\ref{fig:G_B_X}. The intersection graph of the subgraphs $G_j$ and $G'_i$ thus forms the complete bipartite graph $K_{n,m}$.

\begin{figure}
    \centering
\begin{tikzpicture}[
    scale=0.7,
    dot/.style={circle, fill=black, inner sep=1.5pt},
    sq/.style={draw=black, line width=1pt, minimum size=7pt, inner sep=0pt},
    bluearr/.style={->, cyan!80!blue, line width=1.5pt},
    greenarr/.style={->, green!60!black, line width=1.5pt},
    pants/.style={draw=red!85!black, line width=1.4pt},
    lab/.style={purple}
]
\foreach \y/\j in {
    7.05/1,
    6.35/2,
    5.65/3
}{
    \node[dot] at (-0.2,\y) {};
    \node[left] at (-0.35,\y) {$v_{\j,1}^L$};
}

\foreach \y/\j in {
    4.45/1,
    3.75/2,
    3.05/3
}{
    \node[dot] at (-0.2,\y) {};
    \node[left] at (-0.35,\y) {$v_{\j,2}^L$};
}

\foreach \y/\j in {
    1.85/1,
    1.15/2,
    0.45/3
}{
    \node[dot] at (-0.2,\y) {};
    \node[left] at (-0.35,\y) {$v_{\j,3}^L$};
}
\foreach \i/\col in {0/cyan,2/green!70!black}{
    \node[sq, fill=\col] at (2.85,6.8-0.35*\i) {};
}
\foreach \y in {7.05,6.35,5.65}\foreach \i in {6.8,6.1}{
    \draw[-,] (-0.2,\y) -- (2.85,\i) {};
}
\node[lab] at (1.6,5.35) {$G_1$};

\foreach \y in {4.45,3.75,3.05}{
    \node[dot] at (-0.2,\y) {};
}
\foreach \i/\col in {0/cyan,2/green!70!black}{
    \node[sq, fill=\col] at (2.85,4.2-0.35*\i) {};
}
\foreach \y in {4.45,3.75,3.05}\foreach \i in {4.2,3.5}{
    \draw[-,] (-0.2,\y) -- (2.85,\i) {};
}
\node[lab] at (1.6,2.75) {$G_2$};

\foreach \y in {1.85,1.15,0.45}{
    \node[dot] at (-0.2,\y) {};
}
\foreach \i/\col in {0/cyan,2/green!70!black}{
    \node[sq, fill=\col] at (2.85,1.6-0.35*\i) {};
}
\foreach \y in {1.85,1.15,0.45}\foreach \i in {1.6,0.9}{
    \draw[-,] (-0.2,\y) -- (2.85,\i) {};
}
\node[lab] at (1.6,0.15) {$G_3$};

\foreach \i in {0,2,4}{
    \node[sq, fill=cyan] at (5.35,5.8-0.35*\i) {};
}
\foreach \y/\j in {
    5.7/1,
    5.0/2
}{
    \node[dot] at (8.5,\y) {};
    \node[right] at (8.65,\y) {$v_{1,\j}^R$};
}
\foreach \y in {5.7,5.0}\foreach \i in {5.8,5.1,4.4}{
    \draw[-,] (8.5,\y) -- (5.35,\i) {};
}

\foreach \i in {0,2,4}{
    \node[sq, fill=green!70!black] at (5.35,3.05-0.35*\i) {};
}
\foreach \y/\j in {
    2.5/1,
    1.8/2
}{
    \node[dot] at (8.5,\y) {};
    \node[right] at (8.65,\y) {$v_{2,\j}^R$};
}
\foreach \y in {2.5,1.8}\foreach \i in {3.05,2.35,1.65}{
    \draw[-,] (8.5,\y) -- (5.35,\i) {};
}

\node[lab] at (7.2,4.1) {$G'_1$};

\node[lab] at (7.2,1.1) {$G'_2$};




\draw[-,cyan,line width=3pt] (5.35,5.8) -- (2.85,6.8);
\node[above=1mm] at (2.85,6.8) {$c_{1,1}$};
\draw[-,cyan,line width=3pt] (5.35,5.1) -- (2.85,4.2);
\node[above=1mm] at (2.85,4.2) {$c_{2,1}$};
\draw[-,cyan,line width=3pt] (5.35,4.4) -- (2.85,1.6);
\node[above=1mm] at (2.85,1.6) {$c_{3,1}$};

\draw[-,green!70!black,line width=3pt] (5.35,3.2) -- (2.85,6.1);
\node[below=1mm] at (2.85,6.1) {$c_{1,2}$};
\draw[-,green!70!black,line width=3pt] (5.35,2.5) -- (2.85,3.5);
\node[below=1mm] at (2.85,3.5) {$c_{2,2}$};
\draw[-,green!70!black,line width=3pt] (5.35,1.8) -- (2.85,0.9);
\node[below=1mm] at (2.85,0.9) {$c_{3,2}$};

\end{tikzpicture}
    \vspace{0.5cm}
    \caption{The graph $G_{B_{\mathsf{X}}}$ for an $m\times n$ base matrix $B$. The left and right endpoints of each of the colored lines in the middle represent the same vertex. The edge labels on the left are $b_{ij}$, whereas the edge labels on the right are $b_{ij}^{-1}$.}
    \label{fig:G_B_X}
\end{figure}

We fix the following notation for the nodes and edges of $\GX$:
\begin{align*}
    \mbox{Variable nodes in } B\otimes I_n: &\ & v_{k,j}^L & \ &k,j\in[n],\\
    \mbox{Variable nodes in } I_m\otimes B^*: &\ & v_{i,\ell}^R &\ & i,\ell\in[m],\\
    \mbox{Check nodes}:&\ &c_{i,j} &\ & i\in[m],\ j\in[n],\\
    \mbox{Edges in } B\otimes I_n: &\ &v_{k,j}^L\sim c_{i,j} &\ & i\in[m],\ j,k\in[n],\\
    \mbox{Edges in } I_m \otimes B^*: &\ &v_{i,\ell}^R\sim c_{i,j} &\ & i,\ell\in[m],\ j\in[n].
\end{align*}
Oriented from variable nodes to check nodes, the edge $v_{k,j}^L\sim c_{i,j}$ is labeled $b_{ik}=P^{\varepsilon_{ik}}$, and the edge $v_{i,\ell}^R\sim c_{i,j}$ is labeled $b_{\ell j}^{-1}=P^{-\varepsilon_{\ell j}}$, consistent with~\eqref{eq:B_X}. With this notation,
\[
G_j = \GX\big[\{v_{1,j}^L,\dots , v_{n,j}^L\}\cup\{c_{1,j},\dots, c_{m,j}\}\big], \qquad j\in[n],
\]
and
\[
G'_i = \GX\big[\{v_{i,1}^R,\dots , v_{i,m}^R\}\cup\{c_{i,1},\dots, c_{i,n}\}\big], \qquad i\in[m],
\]
so that $G_j\cap G'_i=\{c_{i,j}\}$.

While we have so far exclusively described $G_{H_{\mathsf{X}}}$, the following result shows that the graphs associated to $H_{\mathsf{X}}$ and $H_{\mathsf{Z}}$ are isomorphic. We therefore reference only $G_{H_{\mathsf{X}}}$ in the remainder of the paper for ease of notation, but all results apply to both $G_{H_{\mathsf{X}}}$ and $G_{H_{\mathsf{Z}}}$.

\begin{theorem}\label{thm:graph_iso}
The graphs $G_{H_{\mathsf{X}}}$ and $G_{H_{\mathsf{Z}}}$ are isomorphic.
\end{theorem}

\begin{proof}
The Kronecker product has the property that $A \otimes C = K_1(C \otimes A)K_2$ for some permutation matrices $K_1$ and $K_2$~\cite{henderson1981vec}. We use this property to construct permutation matrices $K$ and $Q$ such that $K B_{\mathsf{X}} Q^T = B_{\mathsf{Z}}$. Let $e_i^{(m)}$ and $e_j^{(n)}$ denote standard basis vectors of length $m$ and $n$, respectively. Define the commutation matrix $K_{m,n}\in\mathbb{F}_2^{mn\times mn}$ as the matrix of the linear map $\mathbb{F}_2^m\otimes \mathbb{F}_2^n \to \mathbb{F}_2^n\otimes \mathbb{F}_2^m$ given by
\[
K_{m,n}\left(e_i^{(m)}\otimes e_j^{(n)}\right)
=
e_j^{(n)}\otimes e_i^{(m)}.
\]
That is, the columns of $K_{m,n}$ are indexed by pairs $(i,j)\in [m]\times [n]$, corresponding to basis vectors
$e_i^{(m)}\otimes e_j^{(n)}$, while the rows are indexed by pairs $(j,i)\in[n]\times [m]$,
corresponding to basis vectors $e_j^{(n)}\otimes e_i^{(m)}$. The entries of $K_{m,n}$ are given by
\[
K_{m,n}(a,b) = \begin{cases}
    1 & \text{if } a=(j,i) \text{ and } b=(i,j),\\
    0 & \text{otherwise}.
\end{cases}
\]
Set $K:=K_{m,n}$ and $Q:=
\begin{bmatrix}
K_{n,n} & \boldzero\\
\boldzero & K_{m,m}
\end{bmatrix}.$
Then
\[
K
\begin{bmatrix}
B\otimes I_n & I_m\otimes B^*
\end{bmatrix}
Q^T
=
\begin{bmatrix}
I_n\otimes B & B^*\otimes I_m
\end{bmatrix}.
\]

To see this, note that right multiplication by $K_{n,n}^T$ permutes the columns of
$B\otimes I_n$ by swapping the two tensor coordinates,
\[
e_i^{(n)}\otimes e_j^{(n)}
\longmapsto
e_j^{(n)}\otimes e_i^{(n)}.
\]
Similarly, right multiplication by $K_{m,m}^T$ permutes the columns of
$I_m\otimes B^*$ by
\[
e_i^{(m)}\otimes e_j^{(m)}
\longmapsto
e_j^{(m)}\otimes e_i^{(m)}.
\]
Left multiplication by $K=K_{m,n}$ permutes the rows by sending the
row label
\[
e_b^{(m)}\otimes e_c^{(n)}
\longmapsto
e_c^{(n)}\otimes e_b^{(m)},
\]
that is, it relabels rows from pairs $(b,c)$ to pairs $(c,b)$. The first block thus satisfies
\[
K_{m,n}(B\otimes I_n)K_{n,n}^T = I_n\otimes B,
\]
and the second block satisfies
\[
K_{m,n}(I_m\otimes B^*)K_{m,m}^T = B^*\otimes I_m.
\]

Since $K$ and $Q$ are permutation matrices, this correspondence between $B_{\mathsf{X}}$ and $B_{\mathsf{Z}}$ lifts to $H_{\mathsf{X}}$ and $H_{\mathsf{Z}}$. Graphically, relabeling the row and column indices is equivalent to relabeling the check nodes and variable nodes, and thus the Tanner graphs of $H_{\mathsf{X}}$ and $H_{\mathsf{Z}}$ are isomorphic.
\end{proof}

As a consequence of Theorem~\ref{thm:graph_iso}, it suffices to analyze just one of the pairs $(G_{B_{\mathsf{X}}}, G_{H_{\mathsf{X}}})$ or $(G_{B_{\mathsf{Z}}}, G_{H_{\mathsf{Z}}})$ when investigating graph properties.

\section{Conditions on Connectivity}\label{section:connectivity}

In this section, we determine the connectivity of the lifted graph $G_{H_{\mathsf{X}}}$ in terms of the permutation products along cycles in the base graph, extending Theorem~\ref{thm:gross_tucker2}. Our first tool is the specialization of Theorem~\ref{thm:gross_tucker1} to $2\times 2$ base matrices. We then give a graph-theoretic description of the relations between the three graphs $G_B$, $G_{B_{\mathsf{X}}}$, and $G_{H_{\mathsf{X}}}$, showing that $\tilde{G}_B$ is connected if and only if $G_{H_{\mathsf{X}}}$ is, and characterize when this happens in terms of permutation products.

\begin{lemma}\label{lem:2x2_connectivity}
Consider a $2 \times 2$ base matrix
$$B= \begin{bmatrix}
    b_{11} & b_{12}\\
    b_{21} & b_{22}\\
\end{bmatrix}$$
with entries $b_{ij}=P^{\varepsilon_{ij}}$, $\varepsilon_{ij} \in \{0,1, \dots, r-1\}$, where $P$ is the generator of $C_r$.

The graph lift $\tilde{G}_B$ is connected if and only if $\gcd(\beta,r)=1$, where
\begin{equation}\label{eq:beta}
\beta := \varepsilon_{11} - \varepsilon_{12} - \varepsilon_{21} + \varepsilon_{22} \pmod{r}.
\end{equation}
\end{lemma}

\begin{proof}
The base graph $G_B=K_{2,2}$ is a $4$-cycle, so $\tilde{G}_B$ consists only of the preimage of this cycle. By Theorem~\ref{thm:gross_tucker1}, this preimage has $\frac{r}{s}$ connected components, where $s$ is the order of the permutation product $b_{11}b_{12}^{-1}b_{22}b_{21}^{-1}=P^\beta$.

The order of $P^\beta$ is the smallest integer $s$ such that $P^{\beta s}=I$, in other words such that $\beta s$ is a multiple of $r$. This number equals $r$ if and only if $\gcd(r,\beta)=1$, which is thus precisely when the number $\frac{r}{s}$ of connected components equals $1$.
\end{proof}

We extend this notion of connectivity to the graph lift $G_{H_{\mathsf{X}}}=\tilde{G}_{B_{\mathsf{X}}}$. In order to apply Theorems~\ref{thm:gross_tucker1} and~\ref{thm:gross_tucker2}, we study permutation products along cycles in $G_{B_{\mathsf{X}}}$ by relating them to cycles in $G_B$ in two different ways.

To establish this correspondence, given a Tanner graph $G$ with edges labeled by elements of $C_r$, we construct an auxiliary graph $G^\circ$ by adding a loop labeled by the identity permutation $I$ at each of its check nodes. We denote the variable and check nodes of $G_B$ by $v_1,\dots,v_n$ and $c_1,\dots,c_m$, respectively. Since $G_{B^*}$ has the same nodes as $G_B$ with the roles of variable and check nodes interchanged, $G_{B^*}^\circ$ is obtained from $G_{B^*}$ by adding a loop at each $v_k$, $k\in[n]$.

\begin{figure}
    \centering
\begin{tikzpicture}[
    scale=0.6,
    dot/.style={circle, fill=black, inner sep=1.5pt},
    sq/.style={draw=black, line width=1pt, minimum size=7pt, inner sep=0pt},
    bluearr/.style={->, cyan!80!blue, line width=1.5pt},
    greenarr/.style={->, green!60!black, line width=1.5pt},
    pants/.style={draw=red!85!black, line width=1.4pt},
    lab/.style={purple},
    every loop/.style={},
]
\foreach \y in {7.05,6.35,5.65}{
    \node[dot] at (-0.2,\y) {};
}
\foreach \i/\col in {0/cyan,2/green!70!black}{
    \node[sq, fill=\col] at (2.85,6.8-0.35*\i) {};
}
\foreach \y in {7.05,6.35,5.65}\foreach \i in {6.8,6.1}{
    \draw[-,] (-0.2,\y) -- (2.85,\i) {};
}

\foreach \y in {4.45,3.75,3.05}{
    \node[dot] at (-0.2,\y) {};
}
\foreach \i/\col in {0/cyan,2/green!70!black}{
    \node[sq, fill=\col] at (2.85,4.2-0.35*\i) {};
}
\foreach \y in {4.45,3.75,3.05}\foreach \i in {4.2,3.5}{
    \draw[-,] (-0.2,\y) -- (2.85,\i) {};
}

\foreach \y in {1.85,1.15,0.45}{
    \node[dot] at (-0.2,\y) {};
}
\foreach \i/\col in {0/cyan,2/green!70!black}{
    \node[sq, fill=\col] at (2.85,1.6-0.35*\i) {};
}
\foreach \y in {1.85,1.15,0.45}\foreach \i in {1.6,0.9}{
    \draw[-,] (-0.2,\y) -- (2.85,\i) {};
}

\foreach \y in {-2,-2.9,-3.8}{
    \node[dot] at (-1,\y) {};
}
\foreach \z in {-2.45,-3.35}{
    \node[dot] (A) at (2.85,\z) {};
    \draw [-] (A) edge [distance=1cm, in=-30, out=30] (A);
}
\foreach \y in {-2,-2.9,-3.8}\foreach \z in {-2.45,-3.35}{
    \draw[-,] (-1,\y) -- (2.85,\z) {};
}

\draw[->](1.5,-0.2)--(0.5,-1.2);
\node[lab] at (1.5,-1.2) {$\Phi_L$};

\draw[->](6.6,-0.2)--(7.6,-1.2);
\node[lab] at (6.6,-1.2) {$\Phi_R$};

\foreach \y in {-2,-2.9,-3.8}{
    \node[dot] (A) at (5.35,\y) {};
        \draw [-] (A) edge [distance=1cm, in=150, out=210] (A);
}
\foreach \z in {-2.45,-3.35}{
    \node[dot] (A) at (9.1,\z) {};
}
\foreach \y in {-2,-2.9,-3.8}\foreach \z in {-2.45,-3.35}{
    \draw[-,] (5.35,\y) -- (9.1,\z) {};
}
\node[lab] at (0.5,-4.5) {$G_B^\circ$};
\node[lab] at (7.5,-4.5) {$G_{B^*}^\circ$};

\foreach \i in {0,2,4}{
    \node[sq, fill=cyan] at (5.35,5.8-0.35*\i) {};
}
\foreach \y in {5.7,5.0}{
    \node[dot] at (8.5,\y) {};
}
\foreach \y in {5.7,5.0}\foreach \i in {5.8,5.1,4.4}{
    \draw[-,] (8.5,\y) -- (5.35,\i) {};
}

\foreach \i in {0,2,4}{
    \node[sq, fill=green!70!black] at (5.35,3.05-0.35*\i) {};
}
\foreach \y in {2.5,1.8}{
    \node[dot] at (8.5,\y) {};
}
\foreach \y in {2.5,1.8}\foreach \i in {3.05,2.35,1.65}{
    \draw[-,] (8.5,\y) -- (5.35,\i) {};
}




\draw[-,cyan,line width=3pt] (5.35,5.8) -- (2.85,6.8);
\draw[-,cyan,line width=3pt] (5.35,5.1) -- (2.85,4.2);
\draw[-,cyan,line width=3pt] (5.35,4.4) -- (2.85,1.6);

\draw[-,green!70!black,line width=3pt] (5.35,3.2) -- (2.85,6.1);
\draw[-,green!70!black,line width=3pt] (5.35,2.5) -- (2.85,3.5);
\draw[-,green!70!black,line width=3pt] (5.35,1.8) -- (2.85,0.9);

\end{tikzpicture}
\vspace{0.5cm}
    \caption{The projections $\Phi_L:\GX\to G_B^\circ$ and $\Phi_R:\GX\to G_{B^*}^\circ$.}
    \label{fig:G_B_X_proj}
\end{figure}

\begin{theorem}\label{thm:projectcycles}
    There are two graph homomorphisms $\Phi_L: \GX\to G_B^\circ$ and $\Phi_R: \GX\to G_{B^*}^\circ$, given by
\begin{align*}
\Phi_L: v_{k,j}^L&\mapsto v_k, & \Phi_R: v_{k,j}^L&\mapsto v_j,\\
c_{i,j}&\mapsto c_i, & c_{i,j}&\mapsto v_j,\\
v_{i,\ell}^R&\mapsto c_i, & v_{i,\ell}^R&\mapsto c_\ell.
\end{align*}
    If $\gamma$ is a closed walk in $\GX$, then the permutation product along $\gamma$ equals the product of the permutation products along the closed walks $\Phi_R(\gamma)$ and $\Phi_L(\gamma)$.
\end{theorem}
\begin{proof}
    The map $\Phi_L$ sends edges in $G_j$ to edges in $G_B^\circ$ for $j\in[n]$, and edges in $G'_i$ to loops in $G_B^\circ$ for $i\in[m]$. Since every edge in $\GX$ lies in exactly one such subgraph, $\Phi_L$ is a homomorphism. The proof that $\Phi_R$ is a homomorphism is identical, with the roles of $G_j$ and $G'_i$ reversed. Moreover, $\Phi_L$ preserves edge labels on each $G_j\cong G_B$ and $\Phi_R$ preserves edge labels on each $G'_i\cong G_{B^*}$, whereas all other edges are sent to loops labeled $I$.

    Let $\gamma$ be a closed walk in $\GX$. Each edge of $\gamma$ has its label preserved by exactly one of the two maps $\Phi_R$ and~$\Phi_L$. Since $C_r$ is abelian, the product of all the labels thus satisfies $$
    \prod_{e\in E(\gamma)}\tau_{\Phi_R(e)} \prod_{e\in E(\gamma)}\tau_{\Phi_L(e)} = \prod_{e\in E(\gamma)}\tau_{\Phi_R(e)} \tau_{\Phi_L(e)} = \prod_{e\in E(\gamma)}I\,\tau_{e} = \prod_{e\in E(\gamma)}\tau_{e},$$ which concludes the proof.
\end{proof}

The intuition behind Theorem~\ref{thm:projectcycles} is that $\Phi_L$ projects the ``left part'' of $G_{B_{\mathsf{X}}}$ to $G_B$, while contracting the ``right part'' to points, where left and right are as in Figure~\ref{fig:G_B_X}. Similarly, $\Phi_R$ projects the ``right part'' of $G_{B_{\mathsf{X}}}$ to $G_{B^*}$, while contracting the ``left part'' to points.

Consider an arbitrary connected directed graph $H$ whose edges are labeled by permutations in $C_r$. We denote by $\Gamma(H)$ the subgroup of $C_r$ of permutation products along closed walks in $H$, noticing that this is indeed a group by the commutativity of $C_r$ and the connectedness of $H$. Indeed, if $\gamma$ and $\gamma'$ are two closed walks and $\pi$ is a path from the base point of $\gamma$ to that of $\gamma'$, then $\gamma\pi\gamma'\pi^{-1}$ is a closed walk whose permutation product is the product of those of $\gamma$ and $\gamma'$.

Clearly, this subgroup is generated by permutation products along cycles, and indeed by permutation products along fundamental cycles with respect to a spanning tree, so Theorem~\ref{thm:gross_tucker2} can be applied.

\begin{theorem}\label{thm:subgroups}
Let $B$ be an $m\times n$ base matrix with entries in $\{0\}\cup C_r$ and $G_B$ connected, and let $B_{\mathsf{X}}=[B\otimes I_n\ |\ I_m\otimes B^*]$. Then $\Gamma(\GX)=\Gamma(G_B)$.
\end{theorem}

\begin{proof}
    Since every cycle in $G_B$ lifts to a cycle in $\GX$ with the same permutation product, it is clear that $\Gamma(G_B)\subseteq \Gamma(\GX)$.

    On the other hand, for any closed walk $\gamma$ in $\GX$, its permutation product can be written as the product of two permutations in $\Gamma(G_B)$ by Theorem~\ref{thm:projectcycles}, namely the permutation product associated to $\Phi_L(\gamma)$ and $\Phi_R(\gamma)$. So every element in $\Gamma(\GX)$ is the product of two elements in $\Gamma(G_B)$, and is therefore itself in the group $\Gamma(G_B)$. This proves that $\Gamma(G_B)\subseteq \Gamma(\GX)\subseteq\Gamma(G_B)$, so equality holds.
\end{proof}

\begin{theorem}\label{thm:connected_iff}
The graph $\tilde{G}_{B_{\mathsf{X}}}=G_{H_{\mathsf{X}}}$ is connected if and only if $\tilde{G}_B$ is.
\end{theorem}
\begin{proof}
    This follows directly from Theorem~\ref{thm:subgroups} and Theorem~\ref{thm:gross_tucker2}, as both graphs are connected if and only if $\Gamma(G_B) = \Gamma(\GX)$ is a transitive subgroup of $S_r$.
\end{proof}

\begin{theorem}\label{thm:iff}
Let $B$ be an $m \times n$ base matrix over $C_r$ with no zero entries, and set $\varepsilon_{i1}=\varepsilon_{j1}=0$ for all $i\in [m],j\in [n]$. That is, the first row and column of $B$ will have identity permutations. Then $\Gamma(G_B) = \langle P^{\varepsilon_{ij}} : 2 \leq i \leq m, 2 \leq j \leq n \rangle$, and $G_{H_{\mathsf X}}$ is connected if and only if $\text{gcd} \left( r, \{ \varepsilon_{ij}\}_{i, j \geq 2} \right) =1$.     
\end{theorem}

\begin{proof}
The identity labeled edges corresponding to the first row and column of $B$ form a spanning tree of $G_B$. Each remaining edge then contributes to a fundamental cycle $v_j\, w_1\, v_1\, w_i \, v_j$ whose permutation product is $P^{-\varepsilon_{ij}}$. Then the subgroup generated by these fundamental cycles is $\langle P^{\text{gcd} \left( r, \{ \varepsilon_{ij}\}_{i, j \geq 2} \right)}\rangle $, and hence, by Theorems \ref{thm:gross_tucker2} and \ref{thm:connected_iff}, we have the result. 
\end{proof}

\begin{proposition} \label{mxn_c_r_connectivity}
Let $B$ be an $m \times n$ base matrix over $C_r$. Suppose $G_B$ contains a $4$-cycle whose permutation product is $P^\delta$, for $$\delta= \varepsilon_{i_1j_1}-\varepsilon_{i_1j_2}-\varepsilon_{i_2j_1} +\varepsilon_{i_2j_2}.$$ If $\gcd(\delta, r)=1$, then $G_{H_{\mathsf{X}}}$ is connected.
\end{proposition}

\begin{proof}
The $4$-cycle is a closed walk in $G_B$, so $P^\delta\in\Gamma(G_B)$.
Since $\gcd(\delta,r)=1$, $P^\delta$ generates $C_r$. Hence
$\Gamma(G_B)=C_r$, and $\tilde G_B$ is connected by
Theorem~\ref{thm:gross_tucker2}. By Theorem~\ref{thm:connected_iff},
$G_{H_{\mathsf X}}$ is connected.
\end{proof}

\begin{example}\label{2x2_H_X_connectivity}
The graph $G_{H_{\mathsf{X}}}$ arising from $2 \times 2$ $B$ over $C_2$ is disconnected if and only if $\beta=0$, for $\beta$ as defined in \eqref{eq:beta}. To see this, we choose a particular spanning tree of $G_{B_{\mathsf{X}}}$, highlighted in Figure~\ref{fig:B_X}, and use the enumeration of vertices given in Section~\ref{section:preliminaries}.

\begin{figure}[ht!]
    \centering
    \begin{tikzpicture}[thick,scale=.95,
    var/.style={circle,draw=black,minimum size=6mm},
    chk/.style={rectangle,draw=black,minimum size=6mm},
    elab/.style={font=\tiny,fill=white,inner sep=1pt,sloped}
]

\node[var,label={left:\small{$v^L_{1,1}$}}] (x11) at (-3.5, 2) {};
\node[var,label={left:\small{$v^L_{2,1}$}}] (x21) at (-3.5, 0.7) {};
\node[var,label={left:\small{$v^L_{1,2}$}}] (x12) at (-3.5,-0.7) {};
\node[var,label={left:\small{$v^L_{2,2}$}}] (x22) at (-3.5,-2) {};

\node[chk] (c11) at (0, 2) {$c_{1,1}$};
\node[chk] (c21) at (0, 0.7) {$c_{2,1}$};
\node[chk] (c12) at (0,-0.7) {$c_{1,2}$};
\node[chk] (c22) at (0,-2) {$c_{2,2}$};

\node[var,label={right:\small{$v^R_{1,1}$}}] (y11) at (3.5, 2) {};
\node[var,label={right:\small{$v^R_{1,2}$}}] (y12) at (3.5, 0.7) {};
\node[var,label={right:\small{$v^R_{2,1}$}}] (y21) at (3.5,-0.7) {};
\node[var,label={right:\small{$v^R_{2,2}$}}] (y22) at (3.5,-2) {};

\path[-,line width=2pt] (x11) edge node[pos=.18,above,elab] {} (c11);
\path[-,line width=2pt] (x21) edge node[pos=.18,above,elab] {} (c11);

\path[-] (x11) edge node[pos=.25,above,elab] {$P^{\varepsilon_{21}}$} (c21);
\path[-,line width=2pt] (x21) edge node[pos=.25,below,elab] {} (c21);

\path[-,line width=2pt] (x12) edge node[pos=.18,above,elab] {} (c12);
\path[-,line width=2pt] (x22) edge node[pos=.18,below,elab] {} (c12);

\path[-] (x12) edge node[pos=.25,above,elab] {$P^{\varepsilon_{21}}$} (c22);
\path[-] (x22) edge node[pos=.25,below,elab] {$P^{\varepsilon_{22}}$} (c22);

\path[-,line width=2pt] (c11) edge node[pos=.82,above,elab] {} (y11);
\path[-,line width=2pt] (c11) edge node[pos=.85,below,elab] {} (y12);

\path[-,line width=2pt] (c12) edge node[pos=.85,above,elab] {} (y11);
\path[-] (c12) edge node[pos=.82,below,elab] {$P^{-\varepsilon_{22}}$} (y12);

\path[-,line width=2pt] (c21) edge node[pos=.82,above,elab] {} (y21);
\path[-,line width=2pt] (c21) edge node[pos=.85,below,elab] {} (y22);

\path[-] (c22) edge node[pos=.75,above,elab] {$P^{-\varepsilon_{12}}$} (y21);
\path[-,line width=2pt] (c22) edge node[pos=.82,below,elab] {} (y22);

\end{tikzpicture}
    \vspace{0.5cm}
    \caption{The graph $G_{B_{\mathsf{X}}}$ for a $2\times 2$ base graph, with a spanning tree highlighted.}
    \label{fig:B_X}
\end{figure}

Consider the tree $T$ with vertex set $V(T)=V(G_{B_{\mathsf{X}}})$ and edge set indicated in Figure~\ref{fig:B_X}.

Since $G_{B_{\mathsf{X}}}$ has $16$ edges and $T$ has $11$ edges, $G_{B_{\mathsf{X}}}$ has five fundamental cycles. We enumerate them as follows, where the suppressed parts of the cycles are the unique path from $c_{1,2}$ to $c_{2,2}$ in the tree.
\begin{equation}\label{eq:2x2_fundamental_cycles}
\begin{aligned}
\gamma_1: & c_{2,1} \sim v_{2,1}^L \sim c_{1,1}\sim v_{1,1}^L\sim c_{2,1}\\
\gamma_2: &  c_{1,2} \sim v_{1,1}^R \sim c_{1,1}\sim v_{1,2}^R\sim c_{1,2}\\
\gamma_3: & c_{2,2} \sim v_{2,1}^R \sim c_{2,1}\sim v_{2,2}^R \sim c_{2,2}\\
\gamma_4: & c_{2,2}\sim v_{1,2}^L \sim c_{1,2} \sim\dots\sim c_{2,2}\\
\gamma_5: & c_{2,2}\sim v_{2,2}^L \sim c_{1,2} \sim\dots\sim c_{2,2}\\
\end{aligned}
\end{equation}

The permutation products of $\gamma_1, \gamma_2$, $\gamma_3$, and $\gamma_5$ are each $P^\beta$, and the permutation product of $\gamma_4$ is $I$. Thus, the permutation products of these fundamental cycles generate $C_2$ if and only if $\beta =1$, implying $G_{H_{\mathsf{X}}}$ is disconnected if and only if $\beta =0$.
\end{example}

\section{Absorbing Sets}\label{section:girth_absorbing}

We now analyze additional graph properties of $G_{B_{\mathsf{X}}}$ and $G_{H_{\mathsf{X}}}$. For the absorbing set analysis, even for graphs $G_B$ and $G_{B_{\mathsf{X}}}$ that have labeled edges, we consider the standard absorbing set definition as stated in Definition~\ref{def:absorbing}, since we are concerned with how these absorbing sets lift to absorbing sets for $G_{H_{\mathsf{X}}}$. Moreover, we exclusively consider absorbing sets whose induced graphs are connected, as these include all inclusion-minimal absorbing sets.

We first consider a $2 \times 2$ base matrix $B$. Then the lifted graph $G_{H_{\mathsf{X}}}$ has variable nodes all of degree $2$. For the parity condition in the absorbing set definition to be satisfied, every variable node in the absorbing set must be incident to only even degree check nodes in the induced subgraph. That is, for any absorbing set $\mathcal{A}$ for $G_{H_{\mathsf{X}}}$, all check nodes in the induced subgraph $G_A$ must be of even degree, leading to the following remark.

\begin{remark}
For a $2 \times 2$ base matrix $B$ over $C_r$, all absorbing sets of $G_{H_{\mathsf{X}}}$ are $(a,0)$ absorbing sets for $a \geq \frac{g(G_{H_{\mathsf{X}}})}{2}$.
\end{remark}

\begin{theorem}\label{thm:2x2_girth}
For a $2 \times 2$ base matrix $B$ over $C_2$ such that $G_{H_{\mathsf{X}}}$ is connected, $G_{H_{\mathsf{X}}}$ has girth $8$.
\end{theorem}

\begin{proof}
By Example~\ref{2x2_H_X_connectivity}, $\beta$ defined in~\eqref{eq:beta} must be nonzero. We show that every $4$-cycle in $G_{B_{\mathsf{X}}}$ has nontrivial permutation product.
It follows that each $4$-cycle arising from an isomorphic subgraph of $G_{B}$ or $G_{B^*}$ has nontrivial permutation product, hence, lifts to $8$-cycles.
Moreover, any two copies of $G_B$ intersect in at most one node, so every $4$-cycle of $G_{B_{\mathsf{X}}}$ is contained in one of the copies.

Finally, the base graph $G_{B_{\mathsf{X}}}$ contains no $6$-cycles, hence there are no $6$-cycles in $G_{H_{\mathsf{X}}}$. Therefore, $g(G_{H_{\mathsf{X}}})=8$.
\end{proof}

\begin{corollary}
The minimal absorbing sets for $G_{H_{\mathsf{X}}}$ as in Theorem~\ref{thm:2x2_girth} are of the form $(8,0)$. There are $12$ $8$-cycles in $G_{H_{\mathsf{X}}}$ for all permutation choices such that $G_{H_{\mathsf{X}}}$ is connected.
\end{corollary}

To characterize the absorbing sets of $H_{\mathsf{X}}$ for a general $m \times n$ base matrix $B$, particularly the minimal absorbing sets, we first describe the absorbing sets of $G_{B_{\mathsf{X}}}$ and then describe how they lift to $G_{H_{\mathsf{X}}}$. 

We divide the absorbing sets of $G_{B_{\mathsf{X}}}$ into three categories: those of variable nodes corresponding to columns strictly from $B \otimes I_n$, which we will call ``left type'' absorbing sets, those of variable nodes corresponding to columns strictly from $I_m \otimes B^{*}$, which we will call ``right type'', and those of variable nodes corresponding to columns from both $B \otimes I_n$ and $I_m \otimes B^*$, which we will call ``mixed type.'' For a set $\mathcal{A}$ to be a minimal ``mixed type'' absorbing set, we require that $\mathcal{A}$ does not contain a minimal ``left type'' or ``right type'' absorbing set.

\begin{theorem}\label{thm:min_left_right_absorbing_set}
Consider an $m \times n$ matrix $B$ with entries in $C_r$. The ``left type'' absorbing sets of $G_{B_{\mathsf{X}}}$ are of the form $(2k, 0)$ for $1 \leq k \leq \lfloor \frac{n}{2}\rfloor$, of which there are $n\binom{n}{2k}$ such absorbing sets. The ``right type'' absorbing sets are of the form $(2\ell,0)$ for $1 \leq \ell \leq \lfloor \frac{m}{2}\rfloor$ of which there are $m\binom{m}{2\ell}$ such absorbing sets. The minimal ``left'' and ``right'' type absorbing sets are thus of the form $(2,0)$. Setting $p:= \lfloor \frac{m}{2} \rfloor +1$ and $q:= \lfloor \frac{n}{2}\rfloor +1$, the minimal ``mixed type'' absorbing sets are of the form
\begin{equation}\label{eq:mixed_type}
\left(
p+q\ ,\ p(\lceil \frac{n}{2}\rceil-1)+q\left(\lceil \frac{m}{2} \rceil-1\right)
\right),
\end{equation}
of which there are
\begin{equation}\label{eq:mixed_count}
n^{q}\binom{n}{q}m^{p}\binom{m}{p}
\end{equation}
such minimal ``mixed type'' absorbing sets.
\end{theorem}

\begin{proof}
Since $G_{B\otimes I_n}$ can be viewed as $n$ disconnected copies of $G_B$, ``left type'' absorbing sets lie within a single copy of $G_B$. Since $G_B \cong K_{m,n}$ with labeled edges, the ``left type'' absorbing sets are of the form $(2k,0)$ for $1 \leq k \leq \lfloor \frac{n}{2}\rfloor$, and there are $n \cdot \binom{n}{2k}$ such absorbing sets. In particular, the minimal ``left type'' absorbing sets are $(2,0)$ absorbing sets, and there are $n\cdot \binom{n}{2}$ of them. The proof is analogous for ``right type'' absorbing sets.
To characterize the minimal ``mixed type'' absorbing sets, we demonstrate the existence of ``mixed type'' absorbing sets with parameters as in \eqref{eq:mixed_type} and then show their minimality. 
Suppose $\mathcal{A}=\mathcal{L}\cup \mathcal{R}$ is a minimal ``mixed type'' absorbing set, where $\mathcal{L}$ is the set of left variable nodes in $\mathcal{A}$ and $\mathcal{R}$ is the set of right variable nodes in $\mathcal{A}$. If any of the blocks $G_i$ or $G_j^{'}$ contains more than one element of $\mathcal{A}$, then these immediately form a left or right absorbing set, contradicting the minimality of $\mathcal{A}$. 
Now consider the sets $L=\{i\in[n]: G_i\cap \mathcal{L}\neq\emptyset\}$ and $R=\{i\in[m]: G_i^{'}\cap \mathcal{R}\neq\emptyset\}.$ If $v_{k,j}^L\in \mathcal{L}$, then by design it has $|R|$ neighbors $c_{ij}$ of degree $2$ and $m-|R|$ neighbors of degree $1$. Similarily, if $v_{i,\ell}^{R}\in \mathcal{R}$, then by design it has $|L|$ neighbors $c_{ij}$ of degree $2$ and $n-|L|$ neighbors of degree $1$. $\mathcal{A}$ is thus an absorbing set if and only if $|L|>\frac{n}{2}$ and $|R|>\frac{m}{2}$, and every such absorbing set contains an inclusion minimal one with $|R|=p$ and $|L|=q$. With respect to $\mathcal{A}$, there are now $p(n-q)$ odd degree check nodes with a neighbor on the right, and $q(m-p)$ odd degree check nodes with a neighbor on the left, yielding the parameters in~\eqref{eq:mixed_type}.
Now $L$ can be chosen in $\binom{n}{q}$ ways and the vertices in $\mathcal{L}$ can be chosen from the blocks of $L$ in $n^q$ ways. Similarily, $R$ can be chosen in $\binom{m}{p}$ ways and the vertices in $\mathcal{R}$ can be chosen from the blocks of $R$ in $m^p$ ways, yielding the count in \eqref{eq:mixed_count}. 
\end{proof}

Understanding the minimal absorbing sets of $G_{B_{\mathsf{X}}}$ allows us to analyze the resulting absorbing sets in the lifted graph.
\begin{corollary}\label{cor:absorbing_set_lift}
Let $\mathcal{A}$ be a $(2,0)$ absorbing set in $G_{B_{\mathsf{X}}}$, and let $\gamma$ denote the corresponding cycle in the induced subgraph $G_{\mathcal{A}}$, with permutation product $\omega(\gamma)=P^{\delta}$. The set $\mathcal{A}$ lifts to a $(2\ord(\omega(\gamma)),0)$ absorbing set in $G_{H_{\mathsf{X}}}$, where $\ord(\omega(\gamma))=\frac{r}{\gcd(r,\delta)}$.
\end{corollary}

\begin{proof}
Given the cycle $\gamma$ in the induced subgraph $G_{\mathcal{A}}$, a walk along $\gamma$ lifts to a closed walk exactly when it traverses the base cycle $\ord(\omega(\gamma))$ times. Moreover, in this case, each check node in this lifted cycle is again incident to exactly two variable nodes, so we obtain an $(a',0)$ absorbing set for $a'=2\ord(\omega(\gamma))$, where $a' \geq 4$ precisely when the connectivity conditions are satisfied.
\end{proof}

\begin{remark}
Assuming the connectivity conditions of Section~\ref{section:connectivity} for an $m \times n$ base matrix $B$ over $C_r$, every $(2,0)$ absorbing set in $G_{B_{\mathsf{X}}}$ lifts to an $(a',0)$ absorbing set for $a' \geq 4$.
\end{remark}

\section{Cycle Analysis}\label{section:cycle_analysis_B}

We turn our attention to describing suitable choices of lift size and permutation entries for a given $m \times n$ base matrix $B$. We wish to choose the entries of $B$ according to the following hierarchy: ensuring connectivity of $G_{H_{\mathsf{X}}}$, maximizing $g(G_{H_{\mathsf{X}}})$, and minimizing the number of short cycles of $G_{H_{\mathsf{X}}}$.

Since the conditions for the connectivity of $G_{H_{\mathsf{X}}}$ depend on analyzing the fundamental cycles of $G_{B_{\mathsf{X}}}$, we establish a systematic choice of spanning tree for $G_{B_{\mathsf{X}}}$ by choosing the first row and column of $B$ to have the identity permutation values. The fundamental cycles of $G_{B_{\mathsf{X}}}$ then have permutation product equal to the permutation product of the additional edges added, resulting in easier analysis.

These $m+n-1$ edges of $G_B$ extend to $m^2+n^2 + 2mn -m -n$ edges in $G_{B_{\mathsf{X}}}$. Since the graph $G_{B_{\mathsf{X}}}$ has $m^2+n^2+mn$ vertices, a spanning tree for $G_{B_{\mathsf{X}}}$ has $m^2 + n^2 + mn-1$ edges. We thus take a subset of these $m^2+n^2 + 2mn -m -n$ edges which forms a spanning tree for $G_{B_{\mathsf{X}}}$.

Given an $r$-lift of $B_{\mathsf{X}}$, to ensure connectivity of $G_{H_{\mathsf{X}}}$, the permutation products of the fundamental cycles of $G_{B_{\mathsf{X}}}$ must generate $C_r$. To avoid $4$-cycles in the lift of $G_{B_{\mathsf{X}}}$, we wish to avoid $2 \times 2$ submatrices of $B$ with trivial permutation product when choosing the permutations for the remaining $(m-1)(n-1)$ entries of $B$. We also wish to avoid choices for the entries of $B$ that yield $2 \times 2$ submatrices with trivial permutation product in the construction of $B_{\mathsf{X}}$, even if they don't result in $2 \times 2$ submatrices with trivial permutation product in $B$.

\begin{theorem}\label{thm:min_lift_size}
Let $B$ be an $m \times n$ base matrix whose first row and column consist of identity permutations. If $G_{H_{\mathsf{X}}}$ has girth $8$, then the lift size satisfies $r \geq \max(m,n)$.
\end{theorem}

\begin{proof}
The matrices $\begin{bmatrix}
    I & I\\
    P^a & P^a
\end{bmatrix}$ and $\begin{bmatrix}
    I & P^a\\
    I & P^a
\end{bmatrix}$ have trivial permutation product. To avoid such submatrices in $B$, the rows and columns of the $(m-1)\times (n-1)$ submatrix of $B$ obtained by deleting the first row and column of $B$ must have distinct entries from $C_r\setminus \{I\}$. By the pigeonhole principle, this requires $r \geq \max(m,n)$.
\end{proof}

These matrices are special cases of a submatrix $\begin{bmatrix}
  P^a & P^b\\
  P^c & P^d
\end{bmatrix}$ with trivial permutation product, that is, with $a - b + d - c \equiv 0 \pmod{r}$. Equivalently, we must avoid entries such that $a + d \equiv b + c \pmod{r}$. It remains an open question to analyze these $2 \times 2$ submatrices of $B$ to obtain not just necessary, but sufficient conditions to ensure that $G_{H_{\mathsf{X}}}$ has maximum possible girth of $8$, and the effect of these conditions on the other properties of these graphs that are relevant for decoding.

Our goal is to obtain necessary and sufficient conditions on the permutation assignment in the lifted product construction such that the resulting Tanner graph of $H_{\mathsf{X}}$ has girth exactly~$8$. Since the lifted product construction always contains $8$-cycles, the girth of $G_{H_{\mathsf{X}}}$ is at most $8$, and girth $8$ is attained exactly when all $4$- and $6$-cycles are eliminated.

Writing a walk as an alternating sequence $W=w_0e_0w_1e_1 \dots e_{k-1}w_{k}$ of vertices and edges, $W$ is \emph{tailless} if $e_0 \neq e_{k-1}$, \emph{backtrackless} if $e_i\neq e_{i+1}$ for all $0 \leq i \leq k-2$, and \emph{closed} if $w_0=w_k$. A tailless backtrackless closed walk is abbreviated as a \emph{TBC walk}. Works using TBC walks as predictors of cycles in protograph-based codes include~\cite{6168186, 6336765, 10251427}. We use the following results, reformulated as in~\cite{10251427}.

\begin{lemma}\cite{6336765}\label{lemma:TBC_projection}
Let $\tilde{G}$ be an $r$-lift of the protograph $G$ and let $W'$ be a $k$-cycle in $\tilde{G}$. Then $W'$ projects onto a TBC walk $W$ of length $\frac{k}{s}$, where $s \geq 1$ is the order of $\omega(W)$ in $C_r$.
\end{lemma}

\begin{theorem}\cite{10251427}\label{thm:TBC_lift}
Let $\tilde{G}$ be an $r$-lift of the protograph $G$ and let $W$ be a TBC walk of length $k$ in $G$. Then the
inverse image of $W$ in $\tilde{G}$ is the union of $\frac{r}{s}$ TBC walks of
length $ks$, where $s \geq 1$ is the order of $\omega(W)$ in~$C_r$.
\end{theorem}

\begin{lemma}\cite{karimi2012message}\label{lemma:k<2g}
Let $G$ be a graph with girth $g$. Then the set of TBC walks of length $k$ coincides with the set of $k$-cycles if $k < 2g$.
\end{lemma}

Combining Lemma~\ref{lemma:TBC_projection}, Theorem~\ref{thm:TBC_lift}, and Lemma~\ref{lemma:k<2g}, we obtain a direct correspondence between the TBC walks in the protograph and the cycles of length $k<2g$ in the lift. In particular, by Theorem~\ref{thm:TBC_lift}, a length-$4$ TBC walk in $G_{B_{\mathsf{X}}}$ with trivial permutation product lifts to $4$-cycles in $G_{H_{\mathsf{X}}}$, and a length-$6$ TBC walk in $G_{B_{\mathsf{X}}}$ with trivial permutation product lifts to $6$-cycles in $G_{H_{\mathsf{X}}}$. Therefore, we seek to classify all length-$4$ and length-$6$ TBC walks, and derive conditions ensuring that every such walk has nontrivial permutation product. A priori, we would also need to consider length-$2$ TBC walks that could lift to $4$-cycles and length-$3$ TBC walks that could lift to $6$-cycles, but since $G_{B_{\mathsf{X}}}$ is simple and bipartite, there are no TBC walks of length $2$ or $3$. Once every length-$4$ and length-$6$ TBC walk has nontrivial permutation product, the lifted Tanner graph has girth $8$, as there are unavoidable $8$-cycles, as discussed in~\cite{raveendran2025minimumdistancesfinitelengthlifted}.

For a $2\times n$ base matrix
\[
B=
\begin{bmatrix}
I & I & \cdots & I\\
I & a_1 & \cdots & a_{n-1}
\end{bmatrix},
\]
whose first row and column consist of identity permutations, we have the following.

\begin{theorem}\label{thm:2xn_4cycles}
The lifted Tanner graph $G_{H_{\mathsf{X}}}$ contains no $4$-cycles if and only if
\[
a_i\neq a_j
\qquad
(0\leq i<j\leq n-1),
\]
where $a_0:=I$.

Equivalently, the second row must consist entirely of distinct permutation labels.
\end{theorem}

This follows because every $2\times2$ submatrix

\[
\begin{bmatrix}
I & I\\
a_i & a_j
\end{bmatrix}
\]

produces the permutation product

\[
a_i^{-1}a_j,
\]

which is trivial exactly when $a_i=a_j$.

Thus distinct entries are both necessary and sufficient for eliminating all lifted $4$-cycles.

\subsection{Classification of Length-$4$ and Length-$6$ TBC Walks in $B_{\mathsf{X}}$}

Let
\[
B=(b_{ij})\in C_r^{m\times n},
\qquad
b_{ij}=P^{\varepsilon_{ij}},
\qquad
\varepsilon_{ij}\in\mathbb Z_r,
\]
and recall that
\[
B_{\mathsf{X}}
=
\left[
B\otimes I_n
\;\middle|\;
I_m\otimes B^*
\right].
\]

We recall the notation for the variable and check nodes from Section~\ref{section:preliminaries}, which we use to identify the length-$4$ and length-$6$ TBC walks.

The check nodes of $G_{B_{\mathsf{X}}}$ are
$ c_{i,j}, i\in[m], j\in[n]$,
the left variable nodes, corresponding to $B\otimes I_n$, are $v^L_{k,j}, k,j\in[n],$
and the right variable nodes, corresponding to $I_m\otimes B^*$, are
$v^R_{i,\ell}, i,\ell\in[m].$

The neighborhoods of the variable nodes are given by $N(v^L_{k,j}) = \{c_{i,j}:i\in[m]\},$ and $N(v^R_{i,\ell}) = \{c_{i,j}:j\in[n]\}.$

Thus, a two-edge walk through a left variable has the form
\[
c_{i,j}
\longrightarrow
v^L_{k,j}
\longrightarrow
c_{i',j},
\]
and therefore preserves the second coordinate of the check nodes.
Similarly, a two-edge walk through a right variable has the form
\[
c_{i,j}
\longrightarrow
v^R_{i,\ell}
\longrightarrow
c_{i,j'},
\]
and therefore preserves the first coordinate of the check nodes. We use this coordinate-preservation property to classify all TBC walks of
length $4$ and $6$.

\subsubsection{Classifying Length-$4$ TBC walks}

\begin{theorem}[Classification of length-$4$ TBC walks]
\label{thm:length4-BX}
For arbitrary $m \times n$ base matrix $B$, the length-$4$ TBC walks of $G_{B_{\mathsf{X}}}$ are either left-type or right-type of the following two forms, up to starting point and orientation:

\begin{enumerate}
    \item \textbf{Left-only:}
    \[c_{i,j} \to v^L_{k,j} \to c_{i',j} \to v^L_{k',j}\to c_{i,j}, \]
    where
    \[i\neq i', \qquad k\neq k'. \]

    \item \textbf{Right-only:}
    \[ c_{i,j} \to v^R_{i,\ell} \to c_{i,j'} \to v^R_{i,\ell'} \to c_{i,j}, \]
    where
    \[j\neq j',\qquad\ell\neq\ell'.\]
\end{enumerate}

In particular, there are no mixed-type length-$4$ TBC walks in $G_{B_{\mathsf{X}}}$.
\end{theorem}

\begin{proof}
Let
\[
W=
c_0\to v_0\to c_1\to v_1\to c_0
\]
be a length-$4$ TBC walk in $G_{B_{\mathsf{X}}}$. Since $W$ is backtrackless,
$c_0\neq c_1$ and $v_0\neq v_1.$ Write $c_0=c_{i,j}$ and $c_1=c_{i',j'}.$ Suppose first that $v_0$ is a left variable node. Since a length 2 walk through a left variable preserves the second coordinate of the check nodes, we must have $j=j'.$
Hence, $c_0=c_{i,j}$ and $c_1=c_{i',j},$ and, since $c_0\neq c_1$, $i\neq i'.$ If $v_1$
were a right variable node, then a length 2 walk $c_{i',j}\to v_1\to c_{i,j}$ through a right variable would preserve the first coordinate. Consequently, $i'=i$, contradicting $i\neq i'$. Therefore, for some $k,k'\in[n]$, $v_0=v^L_{k,j}$ and $v_1=v^L_{k',j}$. Since $v_0\neq v_1$, we have $k\neq k'$. Hence $W$ has the form
\[c_{i,j}\to v^L_{k,j} \to c_{i',j} \to v^L_{k',j} \to c_{i,j}.\] The argument is symmetric if $v_0$ is a right variable node. In particular, no mixed length-$4$ TBC walk exists.
\end{proof}

\begin{corollary}[Permutation conditions for length $4$]
\label{cor:length4-BX}
A left-only length-$4$ TBC walk
\[c_{i,j} \to v^L_{k,j} \to c_{i',j} \to v^L_{k',j} \to c_{i,j}\]
has trivial permutation product if and only if
\[\varepsilon_{ik} - \varepsilon_{i'k} + \varepsilon_{i'k'} - \varepsilon_{ik'} \equiv 0 \pmod r. \]

Consequently, $G_{H_{\mathsf{X}}}$ has no $4$-cycles arising from left-only TBC walks if
and only if
\[ \varepsilon_{ik} -\varepsilon_{i'k} +\varepsilon_{i'k'} -\varepsilon_{ik'} \not\equiv0\pmod r\]
for every $i\neq i'$ and $k\neq k'$. The same conditions avoid $4$-cycles arising from right-only walks.
\end{corollary}

\begin{proof}
The four edges of the left-only walk correspond to the entries
\[
b_{ik},\quad
b_{i'k},\quad
b_{i'k'},\quad
b_{ik'}.
\]
Accounting for the alternating orientation of the traversal, the
permutation product is
\[
b_{ik}^{-1}b_{i'k}b_{i'k'}^{-1}b_{ik'}
\]
or its inverse, depending on the chosen orientation and starting
vertex. Since the entries belong to the cyclic group $C_r$,
this product is trivial precisely when
\[
-\varepsilon_{ik}
+\varepsilon_{i'k}
-\varepsilon_{i'k'}
+\varepsilon_{ik'}
\equiv0\pmod r,
\]
equivalently,
\[
\varepsilon_{ik}
-\varepsilon_{i'k}
+\varepsilon_{i'k'}
-\varepsilon_{ik'}
\equiv0\pmod r.
\]

The right-only walks are the corresponding $4$-cycle patterns in
$B^*$. Since the entries of $B^*$ are inverses of the transposed
entries of $B$, they yield the same collection of vanishing
conditions.
\end{proof}

\subsubsection{Classifying length-$6$ TBC walks}

We now classify the length-$6$ TBC walks.

\begin{lemma}
\label{lem:length6-simple}
Every length-$6$ TBC walk in the simple bipartite graph
$G_{B_{\mathsf{X}}}$ consists of three distinct check nodes and three
distinct variable nodes.
\end{lemma}

\begin{proof}
Write the walk as
\[
W=
c_0
\to v_0
\to c_1
\to v_1
\to c_2
\to v_2
\to c_0.
\]

Backtracklessness gives
\[
c_0\neq c_1,\qquad
c_1\neq c_2,\qquad
c_2\neq c_0.
\]
Indeed, equality of any two consecutive check nodes would force the
intervening edge to be immediately traversed in reverse.

Similarly,
\[
v_0\neq v_1,\qquad
v_1\neq v_2.
\]
Finally, $v_2\neq v_0$ follows from the tailless condition: if
$v_2=v_0$, then the final edge of the walk is the same edge as the
first edge, traversed in the opposite direction.

Hence the three check nodes and three variable nodes are pairwise
distinct.
\end{proof}

\begin{theorem}[Classification of length-$6$ TBC walks]
\label{thm:length6-BX}
Every length-$6$ TBC walk in $G_{B_{\mathsf{X}}}$ is either left-only
or right-only. There are no mixed length-$6$ TBC walks.

More precisely, every length-$6$ TBC walk is, up to choice of starting
point and orientation, of exactly one of the following forms.

\begin{enumerate}

    \item \textbf{Left-only:}
    \[
    \begin{aligned}
    c_{i,j}
    &\to v^L_{k,j}
    \to c_{i',j}
    \to v^L_{k',j}\\
    &\to c_{i'',j}
    \to v^L_{k'',j}
    \to c_{i,j},
    \end{aligned}
    \]
    where
    \[
    i,i',i''\in[m]
    \]
    are pairwise distinct and
    \[
    k,k',k''\in[n]
    \]
    are pairwise distinct.

    \item \textbf{Right-only:}
    \[
    \begin{aligned}
    c_{i,j}
    &\to v^R_{i,\ell}
    \to c_{i,j'}
    \to v^R_{i,\ell'}\\
    &\to c_{i,j''}
    \to v^R_{i,\ell''}
    \to c_{i,j},
    \end{aligned}
    \]
    where
    \[
    j,j',j''\in[n]
    \]
    are pairwise distinct and
    \[
    \ell,\ell',\ell''\in[m]
    \]
    are pairwise distinct.

\end{enumerate}

Consequently, length-$6$ TBC walks exist only if
\[
m\geq3
\qquad\text{and}\qquad
n\geq3.
\]
\end{theorem}

\begin{proof}
Let
\[
W=
c_0
\to v_0
\to c_1
\to v_1
\to c_2
\to v_2
\to c_0
\]
be a length-$6$ TBC walk.

By Lemma~\ref{lem:length6-simple}, the three check nodes are pairwise
distinct.

Each variable node determines a two-edge step between two consecutive
check nodes. There are exactly two possible types of such steps: a step through a left variable $v^L_{k,j}$ that preserves the second coordinate of the check nodes, or a step through a right variable $v^R_{i,\ell}$ that preserves the first coordinate of the check nodes.

We claim that the three two-edge steps cannot contain both types. Suppose, without loss of generality, that the first step is left.
Write $c_0=c_{i,j}.$
Since
\[
c_0\to v_0\to c_1
\]
is a left step, it preserves the second coordinate, and therefore $c_1=c_{i',j}$
for some $i'\neq i$.

Suppose the second step is right. Since
\[
c_1\to v_1\to c_2
\]
is a right step, it preserves the first coordinate, so $c_2=c_{i',j'}$
for some $j'\neq j$.

The final two-edge step would have to connect
\[
c_2=c_{i',j'}
\qquad\text{to}\qquad
c_0=c_{i,j}.
\]
It cannot be a left step, because a left step requires equality of
the second coordinates, whereas $j'\neq j.$
It cannot be a right step, because a right step requires equality of
the first coordinates, whereas $i'\neq i.$
This is impossible.

The same contradiction occurs for any cyclic arrangement containing
both a left and a right step. Hence all three two-edge steps have the
same type. Suppose first that all three are left. Since every left step preserves
the second coordinate, there is a fixed $j\in[n]$ such that
\[
c_0=c_{i,j},
\qquad
c_1=c_{i',j},
\qquad
c_2=c_{i'',j}.
\]
The check nodes are pairwise distinct, so $ i,\ i',\ i''$
are pairwise distinct.

Likewise the three variable nodes have the form $v^L_{k,j}, v^L_{k',j}, v^L_{k'',j}.$
By Lemma~\ref{lem:length6-simple}, these variable nodes are pairwise
distinct, so $k,\ k',$ and $k''$ are pairwise distinct. Thus every left-only length-$6$ TBC walk has the stated form.

Conversely, any walk of the stated form with pairwise distinct
$i,i',i''$ and pairwise distinct $k,k',k''$ is closed,
backtrackless, and tailless, and hence is a TBC walk.

The right-only case is analogous. All three right steps preserve the
first check coordinate, yielding $c_{i,j},
c_{i,j'},$ and $c_{i,j''}$ with $j,j',j''$ pairwise distinct, and the three right variable nodes are $v^R_{i,\ell},
v^R_{i,\ell'},$ and $v^R_{i,\ell''}$
with $\ell,\ell',\ell''$ pairwise distinct.

Therefore every length-$6$ TBC walk is either entirely left or
entirely right, and no mixed length-$6$ TBC walk exists.
\end{proof}

\begin{corollary}[Permutation conditions for length $6$]
\label{cor:length6-BX}
A left-only length-$6$ TBC walk
\[
\begin{aligned}
c_{i,j}
&\to v^L_{k,j}
\to c_{i',j}
\to v^L_{k',j}\\
&\to c_{i'',j}
\to v^L_{k'',j}
\to c_{i,j}
\end{aligned}
\]
has trivial permutation product if and only if
\[
\varepsilon_{ik}
-\varepsilon_{i'k}
+\varepsilon_{i'k'}
-\varepsilon_{i''k'}
+\varepsilon_{i''k''}
-\varepsilon_{ik''}
\equiv0\pmod r.
\]

Consequently, all length-$6$-cycles arising from these TBC walks are
eliminated precisely when
\[
\varepsilon_{ik}
-\varepsilon_{i'k}
+\varepsilon_{i'k'}
-\varepsilon_{i''k'}
+\varepsilon_{i''k''}
-\varepsilon_{ik''}
\not\equiv0\pmod r
\]
for every choice of pairwise distinct $
i,i',i''\in[m]$
and pairwise distinct $
k,k',k''\in[n].$

The right-only length-$6$ walks impose the same collection of
conditions.
\end{corollary}

\begin{proof}
The six edges of the left-only walk correspond, in cyclic order, to
the entries
\[
b_{ik},\quad
b_{i'k},\quad
b_{i'k'},\quad
b_{i''k'},\quad
b_{i''k''},\quad
b_{ik''}.
\]
Thus, up to taking the inverse of the entire product, its permutation
product has exponent
\[
\varepsilon_{ik}
-\varepsilon_{i'k}
+\varepsilon_{i'k'}
-\varepsilon_{i''k'}
+\varepsilon_{i''k''}
-\varepsilon_{ik''}.
\]
The product is the identity precisely when this exponent is zero
modulo $r$.

A right-only length-$6$ TBC walk is the corresponding length-$6$
pattern in $B^*$. Since $B^*$ is obtained from $B$ by transposition
and inversion, its zero-product conditions are precisely the same
collection of conditions after relabeling the row and column indices.
\end{proof}

\subsubsection{Classifying length-$8$ TBC walks}

We now turn our attention to classifying the length-$8$ TBC walks.

\begin{theorem}\label{thm:length-8TBC}
Given an $m\times n$ base matrix $B$, the length-$8$ TBC walks in the Tanner graph of $B_{\mathsf{X}}$ are, up to cyclic permutation and reversal, of the following types:

\begin{enumerate}
    \item a left-only length-$8$ TBC walk
    \item a right-only length-$8$ TBC walk
    \item a mixed-type TBC walk of the form $LRLR$. In walk notation,

   \[
   c_{i_1, j_1} \to v_{k_1, j_1}^L \to c_{i_2, j_1} \to v_{i_2, \ell_1}^R \to c_{i_2, j_2} \to v_{k_2, j_2}^L \to c_{i_1, j_2} \to v_{i_1, \ell_2}^R \to c_{i_1, j_1}
   \]

   where we require $i_1 \neq i_2$ and $j_1 \neq j_2$.

    \item a mixed-type TBC walk of the form $LLRR$. In walk notation,

    \[
    c_{i_1,j_1} \to v_{k_1,j_1}^L \to c_{i_2, j_1} \to v_{k_2,j_1}^L \to c_{i_1, j_1} \to v_{i_1, \ell_1}^R \to c_{i_1, j_2} \to v_{i_1, \ell_2}^R \to c_{i_1, j_1}
    \]

    where we require $i_1 \neq i_2$, $j_1 \neq j_2$, $k_1 \neq k_2$, $\ell_1 \neq \ell_2$.
\end{enumerate}
\end{theorem}

In~\cite{raveendran2025minimumdistancesfinitelengthlifted}, the authors noted the occurrence of length-$8$ walks that always have trivial permutation product, thus resulting in unavoidable $8$-cycles. Indeed, given the LRLR length-$8$ TBC walk described in Theorem~\ref{thm:length-8TBC}, setting
\[
k_1=j_1, \qquad k_2=j_2, \qquad \ell_1=i_1, \qquad \ell_2 = i_2,
\]
we obtain the TBC walk
\[
    c_{i_1,j_1} \to v_{j_1,j_1}^L \to c_{i_2, j_1} \to v_{j_2,j_1}^L \to c_{i_1, j_1} \to v_{i_1, i_1}^R \to c_{i_1, j_2} \to v_{i_1, i_2}^R \to c_{i_1, j_1}
\]
with permutation product
\[
b_{i_1j_1}b_{i_2j_1}^{-1}b_{i_1j_1}^{-1}b_{i_1j_2}b_{i_2j_2}b_{i_2j_2}b_{i_1j_2}^{-1}b_{i_2j_2}^{-1}b_{i_2j_1}.
\]
Since $C_r$ is abelian, this product is trivial, resulting in $8$-cycles in the lift. An interesting direction for future work is to choose edge labelings from a nonabelian group to possibly increase the girth.

\subsubsection{Constructions}

Now that we have classified all length-$4$, $6$, and $8$ TBC walks, we summarize the conditions needed to obtain a lifted graph $G_{H_{\mathsf{X}}}$ with girth $8$ for an $m \times n$ base matrix $B$. We can view this as a matrix completion problem, with the goal of determining edge labelings such that $G_{H_{\mathsf{X}}}$ is, first, connected and, second, of girth $8$.

First, choose the first row and column of $B$ to be identity permutations, so that we obtain a spanning tree for $G_{B_{\mathsf{X}}}$ as discussed at the beginning of this section. Next, choose the second row $P^{\varepsilon_{2j}}$, $2 \leq j \leq n$, such that $\langle P^{\varepsilon_{22}}, \dots, P^{\varepsilon_{2n}} \rangle = C_r$, which ensures the connectivity of $G_{H_{\mathsf{X}}}$. Finally, we avoid length-$4$ and length-$6$ TBC walks with trivial permutation product. These conditions are summarized as follows.

\begin{itemize}
    \item Every $2 \times 2$ submatrix of $B$ has nonzero alternating exponent sum (Corollary~\ref{cor:length4-BX}).
    \item Every $3 \times 3$ $6$-cycle pattern of $B$ has nonzero alternating exponent sum (Corollary~\ref{cor:length6-BX}).
\end{itemize}

\section{Simulations}\label{section:simulations}
We now construct explicit examples of finite length codes satisfying the aforementioned mathematical criteria, and compare their decoding performance with other code families. Using the connectivity conditions of Section \ref{section:connectivity}, particularly Theorem \ref{thm:iff}, and length 4 and 6 TBC walk conditions of \ref{section:cycle_analysis_B}, Corollaries \ref{cor:length4-BX} and \ref{cor:length6-BX}, we search for permutation choices of base matrices $B$ that satisfy these conditions and then analyze decoding performance of the resulting constraint-aware choices for $B$.

Namely, we seek a base matrix $B$ with no zero entries such that $G_{H_{\mathsf{X}}}$ is connected and has girth 8. We search over $3 \times 4$ base matrices $B$, as the authors in \cite{raveendran2025minimumdistancesfinitelengthlifted} derive Row-Column Partition Constraints (RCPC) to avoid problematic low-weight logical codewords. They find that if $B$ has $2$ rows or columns, RCPC is always satisfied, resulting in an unavoidable restriction on the minimimum distance. The pseudocode for our search is presented in Algorithm \ref{alg:search}. 

\RestyleAlgo{ruled}
\SetKwComment{Comment}{/* }{ */}
\begin{algorithm}[!htb]
\DontPrintSemicolon
\caption{Search for Admissible labelings}\label{alg:search}
\KwIn{lift size $r$; dimensions $3 \times 4$ for an all-one protograph $B$}
\KwOut{set $\mathcal{A}$ of admissible labelings}
$B_{1j} \gets P^0$ for $1 \le j \le 4$, and $B_{i1} \gets P^0$ for $1 \le i \le 3$\;
$\mathcal{A}\gets\emptyset$\;
\ForEach{$(\varepsilon_{ij})_{i,j\ge 2}\in\mathbb{Z}_r^{6}$}{
$B_{ij} \gets P^{\varepsilon_{ij}}$ for $2\le i\le 3$, $2\le j\le 4$\;
  \lIf(\tcp*[f]{Section \ref{section:connectivity}, Theorem \ref{thm:iff}}){$\gcd\bigl(r,\{\varepsilon_{ij}\}\bigr)\neq 1$}{\Reject}
  \lIf(\tcp*[f]{Section \ref{section:cycle_analysis_B}, Corollary \ref{cor:length4-BX}}){a length-4 TBC walk has product $I$}{\Reject}
  \lIf(\tcp*[f]{Section \ref{section:cycle_analysis_B}, Corollary \ref{cor:length6-BX}}){a length-6 TBC walk has product $I$}{\Reject}
  \lIf(\tcp*[f]{\cite{raveendran2025minimumdistancesfinitelengthlifted}}){$B$ satisfies an RCPC}{\Reject}
  $\mathcal{A}\gets\mathcal{A}\cup\{B\}$\;
}
\Return{$\mathcal{A}$}\;
\end{algorithm}

We illustrate the results of this search in Table \ref{tab:table}, where we select from $\mathcal{A}$ the labeling whose lifted graph has the
fewest $8$-cycles. By ``Label choices for $g=8$'' we mean the number of label choices that yield a connected Tanner graph of girth $8$, not taking the RCPC conditions into account. By the ``admissible'' column, we mean the number of label choices that give constructions that are girth $8$, connected, and avoid the RCPC conditions. By $\text{min}(c_8)$ we mean the minimum number of $8$ cycles across such labeling choices. The parameters $[[n,k]]$ reflect the length and dimension for a select labeling choice achieving $\text{min}(c_8)$. We note that the parameters were the same for each of the labelings achieving $\text{min}(c_8)$ for a given lift size $r$.

\begin{table}[ht]
\begin{center}
\setlength{\tabcolsep}{3pt}
\begin{tabular}{ |c|c|c|c|c|  }
 \hline

 $C_r$ & Label choices for $g=8$& Admissable & $\text{min}(c_8)$ & $[[n,k]]$\\
 \hline
 $C_2- C_8$   & 0    &--&   -- & --\\
 $C_9$ &   216  & 0 & -- & --\\
 $C_{10}$ & 1,368 & 1,152 & 280&  [[250,22]]\\
 $C_{11}$    &4,680 & 3,600 & 264& [[275,23]] \\
 $C_{12}$ &   16,992  & 13,968 & 252& [[300,24]]\\
 $C_{13}$& 48,600 & 40,608 & 247   &[[325, 25]]\\
  $C_{14}$& 116,784 & 101,952 & 231   &[[350, 26]]\\
 \hline
\end{tabular}
\vspace{0.7cm}
\caption{\footnotesize{Search results for $3\times 4$ base matrix $B$}}
\label{tab:table}
\end{center}
\end{table}

Using a depolarizing noise model and Belief Propagation + Ordered Statistics decoding software from \cite{roffe_decoding_2020}, we perform decoding simulations to assess decoder performance using these design criterion. We compare the code obtained from the labeling selected by
Algorithm~\ref{alg:search} at $r=24$, which we call Optimized, for comparable codes constructed from $B^{(1)}_{d^Q_{min}=7}$ and $B^{(1)}_{d^Q_{min}=20}$ from \cite{raveendran2025minimumdistancesfinitelengthlifted}, which we refer to as $B_1$ and $B_2$ in the figure legend, respectively, along with a benchmarking hypergraph product code of comparable length. We see almost identical decoding performances for the optimized code and $B_2$, as well, and significantly better performance between the optimized code and $B_1$, as well as a benchmarking length 625 hypergraph product code \cite{tillich2013quantum}. 

We provide the exponent matrices for the codes in Figures \ref{fig:sim1}, \ref{fig:sim2} in the Table \ref{tab:matrices}.  

\begin{table}[ht]
\centering
\begin{tabular}{|c|c|}
\hline 
Parameters & Matrix\\
\hline 
$r=13$, random &
$\left[\begin{smallmatrix}0&0&0&0\\0&1&10&8\\0&5&5&11
\end{smallmatrix}\right]$\\[6pt]
\hline 
$r=13$, typical ($c_8=260$) &
$\left[\begin{smallmatrix}0&0&0&0\\0&11&5&7\\0&2&3&8
\end{smallmatrix}\right]$\\[6pt]
\hline 
$r=13$, selected ($c_8=247$) &
$\left[\begin{smallmatrix}0&0&0&0\\0&12&7&4\\0&3&2&1
\end{smallmatrix}\right]$\\[6pt]
\hline 
$r=24$, ours ($c_8=156$) &
$\left[\begin{smallmatrix}0&0&0&0\\0&5&20&3\\0&22&23&12
\end{smallmatrix}\right]$\\[6pt]
\hline 
$r=24$, $B1$ of~\cite{raveendran2025minimumdistancesfinitelengthlifted} &
$\left[\begin{smallmatrix}0&0&0&0\\0&1&2&3\\0&5&9&14
\end{smallmatrix}\right]$\\[6pt]
\hline 
$r=24$, $B2$ of~\cite{raveendran2025minimumdistancesfinitelengthlifted} &
$\left[\begin{smallmatrix}0&0&0&0\\0&1&3&7\\0&9&16&5
\end{smallmatrix}\right]$\\
\hline 
\end{tabular}
\vspace{1cm}
\caption{Exponent matrices $(\varepsilon_{ij})$ of the simulated
codes, where $b_{ij}=P^{\varepsilon_{ij}}$.}
\label{tab:matrices}
\end{table}

We simulate code-capacity depolarizing noise and decode the X and Z syndromes separately with BP+OSD software from \cite{roffe_decoding_2020}, using scaled min-sum BP (scaling factor 0.625, at most 100 iterations, serial schedule) followed by combination-sweep OSD of order 10. Figure \ref{fig:sim1} compares three $[[325, 25]]$ codes over $C_{13}$ that differ only in their labeling. Precisely, we choose a random labeling (girth $4$, $d\le 10$), a typical admissible labeling
(girth $8$, $c_8=260$, $d\le 12$), and the selected labeling (girth $8$,
$c_8=247$, $d\le 14$), with distances upper bounds from randomized
information-set search. Imposing the conditions of Algorithm~\ref{alg:search} lowers the logical error rate by more than an
order of magnitude for $p\le 0.04$. The further gain of the selected labeling appears only at low $p$ and coincides with its larger distance bound, so we do not attribute it to the cycle count alone. This example illustrates the improvement that can be found by utilizing our design constraints to construct explicit finite length code constructions.

Distances marked $\le$ are upper bounds computed with a randomized information-set search, namely the Lee--Brickell algorithm~\cite{lee1988observation} with parameter $p=2$, which refining Prange's information-set decoding~\cite{prange1962use}. In
each iteration we apply a uniformly random permutation to the qubits, and bring a basis of $\ker H_{\mathsf X}$ to reduced row-echelon
form, so that its pivot columns form a random information set. Each row of the result has weight $1$ on this information set, and each
sum of two rows has weight $2$. We test every row, and every sum of two rows at least one of which is among the first $40$ rows, for
being a nontrivial $\mathsf Z$-type logical operator. Meaning, a vector of $\ker H_{\mathsf X}$ outside the row space of $H_{\mathsf Z}$. We keep the lowest weight one found. Each reported
weight is the weight of an explicit logical operator, so it is an upper bound on $d_{\mathsf Z}$, but the search does give us a lower bound. By Theorem~\ref{thm:graph_iso}, $d_{\mathsf X}=d_{\mathsf Z}$, so it suffices to search one type. For the $r=13$ codes we ran $400$ iterations for each of $8$ random seeds, and for the $r=24$ codes we ran $1200$ iterations for each of $12$ seeds and we are presenting the minimum over all the seeds. For the $r=13$
codes, every seed returned the same value. Applied to the codes $B_1$ and $B_2$ of~\cite{raveendran2025minimumdistancesfinitelengthlifted}, the search returns the distances $7$ and $20$ from their paper at every seed.

\begin{remark}
Every labeling returned by Algorithm~\ref{alg:search} has no $4$-cycles, so every $\delta_{i_1i_2}\not\equiv 0 \pmod r$. When $r$ is prime, this forces
$S_{k_1k_2}=C_r$, so each $(2,0)$ absorbing set of $G_{B_{\mathsf X}}$ lifts to a single $(2r,0)$-absorbing set of $G_{H_{\mathsf X}}$. For the codes of Fig.~\ref{fig:sim1} these are $(26,0)$-absorbing sets.
\end{remark}

\begin{figure}[ht!]
    \centering
    \includegraphics[width=0.5\linewidth]{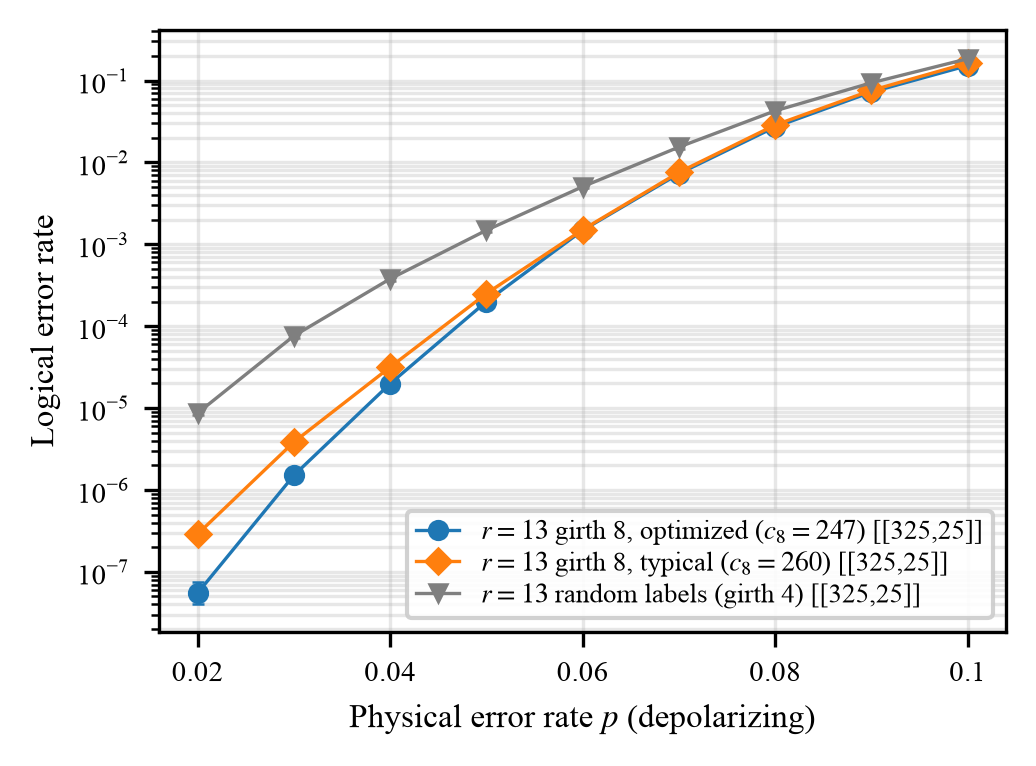}
    \vspace{0.7cm}
    \caption{LER of three $[[325,25]]$ lifted-product codes over $C_{13}$. The girth-8 labelings are produced by Algorithm~\ref{alg:search} for minimum and typical $c_8$, and the random labeling contains 4-cycles.}
    \label{fig:sim1}
\end{figure}

At $r=24$, the same procedure gives a $[[600,36,\le 16]]$ code that we refer to as ``Optimized'' in Fig \ref{fig:sim2}, reflecting that its labeling was chosen using Algorithm \ref{alg:search}. This Optimized code has base lift with $156$ $8$-cycles, fewer than the $168$ and $264$ of the $[[600,36,20]]$ and
$[[600,36,7]]$ codes of~\cite{raveendran2025minimumdistancesfinitelengthlifted}. In our simulations, this Optimized code decodes comparably with the $d=20$ code, even though the minimum distance is smaller. We conjecture that the small number of $8$ cycles is advantageous. At the very least, Algorithm \ref{alg:search} provides a heuristic for choosing labelings to optimize decoding performance for finite length constructions when precise guarantees about minimum distance are unknown. 

\begin{figure}[ht!]
    \centering
    \includegraphics[width=0.5\linewidth]{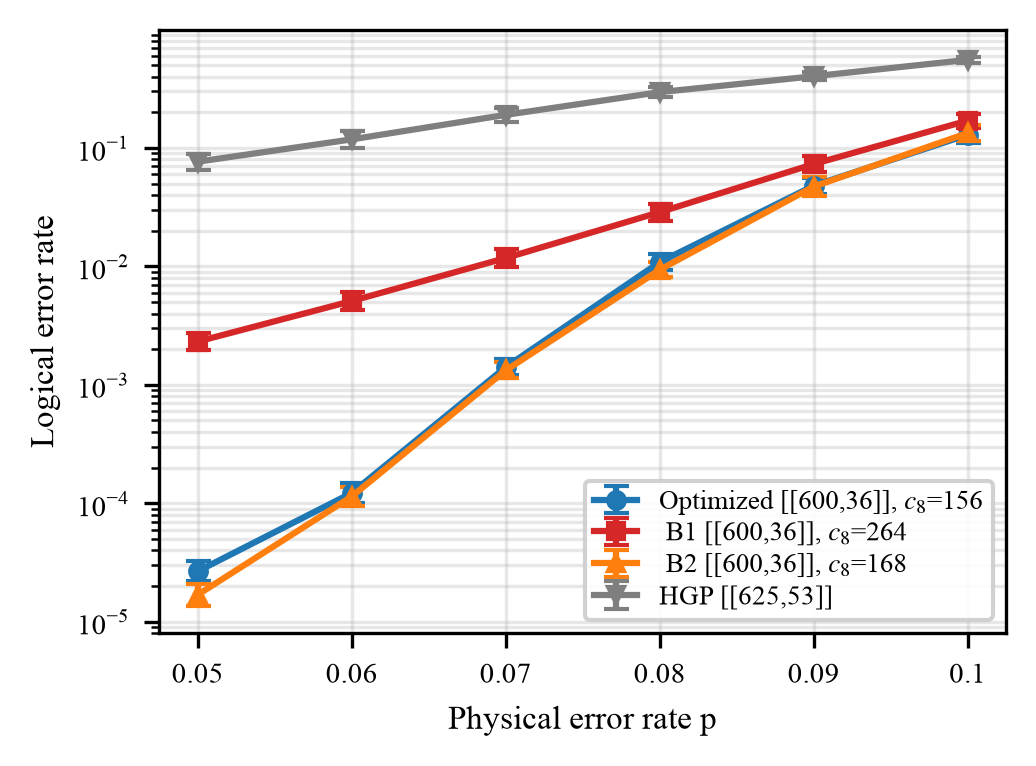}
    \vspace{0.7cm}
    \caption{LER of lifted-product codes over $C_{24}$, comparing $[[600, 36, \leq 16]]$ code produced by Algorithm~\ref{alg:search} with the $[[600,36,20]]$ and $[[600,36,7]]$ codes of~\cite{raveendran2025minimumdistancesfinitelengthlifted} and a benchmarking length 625 hypergraph product code.}
    \label{fig:sim2}
\end{figure}

\section{Conclusion}\label{section:conclusion}
This work analyzes the structure of the Tanner graphs corresponding to $H_{\mathsf{X}}$ and $H_{\mathsf{Z}}$ for a particular case of the lifted product code construction. We show that these graphs are isomorphic and provide conditions ensuring that they are connected, which is necessary for robust graph-based iterative decoding. We characterize and enumerate the minimal absorbing sets of $G_{B_{\mathsf{X}}}$ and describe how these absorbing sets lift to $G_{H_{\mathsf{X}}}$. Finally, we classify the TBC walks of length $4$, $6$, and $8$ in $G_{B_{\mathsf{X}}}$, which yields conditions on the base matrix $B$ under which $G_{H_{\mathsf{X}}}$ has girth $8$. In~\cite{raveendran2025minimumdistancesfinitelengthlifted}, the authors derive row-column partition constraints, which guarantee the existence of low-weight logical operators; it would be interesting to see how these constraints translate to properties of the Tanner graphs. Future work includes further describing the effect of the permutation choices on other graph properties relevant to decoding performance, such as the minimal absorbing sets characterized in this work and expansion quality.

\bibliographystyle{IEEEtran}
\bibliography{references}

@article{Dolecek07,
  title={Analysis of Absorbing Sets for Array-Based {LDPC} Codes},
  author={Lara Dolecek and Zhengya Zhang and Venkat Anantharam and Martin J. Wainwright and Borivoje Nikoli{\'c}},
  journal={2007 IEEE International Conference on Communications},
  year={2007},
  pages={6261-6268}
}

@article{CS96,
  title = {Good quantum error-correcting codes exist},
  author = {Calderbank, A. R. and Shor, Peter W.},
  journal = {Phys. Rev. A},
  volume = {54},
  issue = {2},
  pages = {1098--1105},
  numpages = {0},
  year = {1996},
  month = {aug},
  publisher = {American Physical Society},
  doi = {10.1103/PhysRevA.54.1098},
  url = {https://link.aps.org/doi/10.1103/PhysRevA.54.1098}
}

@ARTICLE{T81,
  author={Tanner, R.},
  journal={IEEE Transactions on Information Theory}, 
  title={A recursive approach to low complexity codes}, 
  year={1981},
  volume={27},
  number={5},
  pages={533-547},
  doi={10.1109/TIT.1981.1056404}}

@article{mackay2004sparse,
  title={Sparse-graph codes for quantum error correction},
  author={MacKay, David JC and Mitchison, Graeme and McFadden, Paul L},
  journal={IEEE Transactions on Information Theory},
  volume={50},
  number={10},
  pages={2315--2330},
  year={2004},
  publisher={IEEE}
}

@article{Gottesman-overhead,
  title={Fault-tolerant quantum computation with constant overhead},
  author={Gottesman, Daniel},
  year={2014},
  month = {nov},
  journal={Quantum Inform. and Computation},
    volume = {14},
  issue = {15-16},
  pages = {1338--1372},
  doi = {https://doi.org/10.26421/QIC14.15-16-5}
}

@article{Morris2026AMC,
title = {Absorbing sets in quantum {LDPC} codes},
journal = {Advances in Mathematics of Communications},
volume = {23},
number = {0},
pages = {234-254},
year = {2026},
issn = {1930-5346},
doi = {10.3934/amc.2026023},
url = {https://www.aimsciences.org/article/id/69a6ac251c9579521d08fdc7},
author = {Kirsten D. Morris and Tefjol Pllaha and Christine A. Kelley}
}

@article{roffe_decoding_2020,
   title={Decoding across the quantum low-density parity-check code landscape},
   volume={2},
   ISSN={2643-1564},
   url={http://dx.doi.org/10.1103/PhysRevResearch.2.043423},
   DOI={10.1103/physrevresearch.2.043423},
   number={4},
   journal={Physical Review Research},
   publisher={American Physical Society (APS)},
   author={Roffe, Joschka and White, David R. and Burton, Simon and Campbell, Earl},
   year={2020},
   month={Dec}
}

@article{tillich2013quantum,
  title={{Quantum {LDPC} codes with positive rate and minimum distance proportional to the square root of the blocklength}},
  author={Tillich, Jean-Pierre and Z{\'e}mor, Gilles},
  journal={IEEE Transactions on Information Theory},
  volume={60},
  number={2},
  pages={1193--1202},
  year={2013},
  publisher={IEEE}
}

@inproceedings{panteleev2022asymptotically,
  title={Asymptotically good quantum and locally testable classical {LDPC} codes},
  author={Panteleev, Pavel and Kalachev, Gleb},
  booktitle={Proceedings of the 54th annual ACM SIGACT symposium on theory of computing},
  pages={375--388},
  year={2022}
}

@INPROCEEDINGS{raveendran2025minimumdistancesfinitelengthlifted,
  author={Raveendran, Nithin and Declercq, David and Vasić, Bane},
  booktitle={ICC 2026 - IEEE International Conference on Communications}, 
  title={On the Minimum Distances of Finite-Length Lifted Product Quantum {LDPC} Codes}, 
  year={2026},
  volume={},
  number={},
  pages={1-6},
  doi={10.1109/ICC59461.2026.11587715}}

@inproceedings{mao2001heuristic,
  title={A heuristic search for good low-density parity-check codes at short block lengths},
  author={Mao, Yongyi and Banihashemi, Amir H},
  booktitle={ICC 2001. IEEE International Conference on Communications. Conference Record (Cat. No. 01CH37240)},
  volume={1},
  pages={41--44},
  year={2001},
  organization={IEEE}
}

@article{halford2005algorithm,
  title={An algorithm for counting short cycles in bipartite graphs},
  author={Halford, Thomas R and Chugg, Keith M},
  journal={IEEE Transactions on Information Theory},
  volume={52},
  number={1},
  pages={287--292},
  year={2005},
  publisher={IEEE}
}

@article{karimi2012efficient,
  title={Efficient algorithm for finding dominant trapping sets of {LDPC} codes},
  author={Karimi, Mehdi and Banihashemi, Amir H},
  journal={IEEE Transactions on Information Theory},
  volume={58},
  number={11},
  pages={6942--6958},
  year={2012},
  publisher={IEEE}
}

@article{mitchell2014quasi,
  title={Quasi-cyclic {LDPC} codes based on pre-lifted protographs},
  author={Mitchell, David GM and Smarandache, Roxana and Costello, Daniel J},
  journal={IEEE Transactions on Information Theory},
  volume={60},
  number={10},
  pages={5856--5874},
  year={2014},
  publisher={IEEE}
}

@article{smarandache2025structural,
  title={Structural Analysis of Generalized Quasi-Cyclic {LDPC} Codes: Rank, Design and Generator Matrices},
  author={Smarandache, Roxana and Mitchell, David GM and G{\'o}mez-Fonseca, Anthony},
  journal={IEEE Transactions on Information Theory},
  year={2025},
  publisher={IEEE}
}

@article{dolecek2014non,
  title={Non-binary protograph-based {LDPC} codes: Enumerators, analysis, and designs},
  author={Dolecek, Lara and Divsalar, Dariush and Sun, Yizeng and Amiri, Behzad},
  journal={IEEE transactions on information theory},
  volume={60},
  number={7},
  pages={3913--3941},
  year={2014},
  publisher={IEEE}
}

@book{gross2001topological,
  title={Topological graph theory},
  author={Gross, Jonathan L and Tucker, Thomas W},
  year={2001},
  publisher={Courier Corporation}
}

@article{henderson1981vec,
  title={The vec-permutation matrix, the vec operator and Kronecker products: A review},
  author={Henderson, Harold V and Searle, Shayle R},
  journal={Linear and multilinear algebra},
  volume={9},
  number={4},
  pages={271--288},
  year={1981},
  publisher={Taylor \& Francis}
}

@ARTICLE{6336765,
  author={Karimi, Mehdi and Banihashemi, Amir H.},
  journal={IEEE Transactions on Communications}, 
  title={Message-Passing Algorithms for Counting Short Cycles in a Graph}, 
  year={2013},
  volume={61},
  number={2},
  pages={485-495},
  doi={10.1109/TCOMM.2012.100912.120503}}

@ARTICLE{6168186,
  author={Asvadi, Reza and Banihashemi, Amir H. and Ahmadian-Attari, Mahmoud},
  journal={IEEE Transactions on Communications}, 
  title={Design of Finite-Length Irregular Protograph Codes with Low Error Floors over the Binary-Input AWGN Channel Using Cyclic Liftings}, 
  year={2012},
  volume={60},
  number={4},
  pages={902-907},
  doi={10.1109/TCOMM.2012.030712.110086}}

@ARTICLE{10251427,
  author={Gómez-Fonseca, Anthony and Smarandache, Roxana and Mitchell, David G. M.},
  journal={IEEE Journal on Selected Areas in Information Theory}, 
  title={An Efficient Strategy to Count Cycles in the Tanner Graph of Quasi-Cyclic LDPC Codes}, 
  year={2023},
  volume={4},
  number={},
  pages={499-513},
  doi={10.1109/JSAIT.2023.3315585}}

@article{karimi2012message,
  title={Message-passing algorithms for counting short cycles in a graph},
  author={Karimi, Mehdi and Banihashemi, Amir H},
  journal={IEEE Transactions on Communications},
  volume={61},
  number={2},
  pages={485--495},
  year={2012},
  publisher={IEEE}
}

@inproceedings{lee1988observation,
author = {Lee, Pil Joong and Brickell, Ernest F.},
title = {An Observation on the Security of {McEliece}’s
Public-Key Cryptosystem},
booktitle = {Advances in Cryptology -- {EUROCRYPT} ’88},
series = {Lecture Notes in Computer Science},
volume = {330},
pages = {275--280},
publisher = {Springer},
year = {1988}
}

@article{prange1962use,
author = {Prange, Eugene},
title = {The Use of Information Sets in Decoding Cyclic Codes},
journal = {IRE Transactions on Information Theory},
volume = {8},
number = {5},
pages = {5--9},
year = {1962}
}

\end{document}